\documentclass[12pt,draftcls,onecolumn]{IEEEtran}
\usepackage{float}
\usepackage{algorithm,color}
\usepackage{algorithmic}
\usepackage{amssymb,amsmath,amsthm,amsopn,amsfonts,graphicx,colordvi,epstopdf,subfigure}
\usepackage[justification=centering]{caption}
 \usepackage{indentfirst}
\input epsf
\def\baselinestretch{1.4}

\def\a{\alpha}
\def\b{\beta}

\def\ba{\begin{array}}
\def\ea{\end{array}}
\def\ban{\begin{eqnarray*}}
\def\ean{\end{eqnarray*}}
\def\bd{\begin{description}}
\def\ed{\end{description}}
\def\be{\begin{equation}}
\def\ee{\end{equation}}
\def\bna{\begin{eqnarray}}
\def\ena{\end{eqnarray}}

\def\d{\delta}

\def\bnaa{\begin{eqnarray}}
\def\enaa{\end{eqnarray}}
\def\bann{\begin{eqnarray*}}
\def\eann{\end{eqnarray*}}

\renewcommand\thesection{\Roman {section}}
\newtheorem{definition}{Definition}[section]
\newtheorem{remark}{Remark}
\newtheorem{theorem}{Theorem}[section]

\newtheorem{lemma}{Lemma}[section]

\newtheorem{corollary}{Corollary}[section]

\begin{document}
\title{Decentralized Cooperative Online Estimation With Random Observation Matrices, Communication Graphs and Time Delays}
\author{Jiexiang Wang,
        Tao Li,~\IEEEmembership{Senior Member,~IEEE},
        Xiwei Zhang
        \thanks{*Corresponding Author: Tao Li. This work is supported by the National Natural Science Foundation of China under Grant 61977024. Please address all the correspondences to Tao Li: Phone: +86-21-54342646-318, Fax: +86-21-54342609, Email: tli@math.ecnu.edu.cn.}
\thanks{ Jiexiang Wang is with School of Mechatronic Engineering and Automation, Shanghai University, Shanghai 200072, China.}
\thanks{ Tao Li and Xiwei Zhang are with the Key Laboratory of Pure Mathematics and Mathematical Practice, School of Mathematical Sciences, East China Normal University, Shanghai 200241, China.}
}

\maketitle

\begin{abstract}
We analyze convergence of  decentralized cooperative online estimation algorithms by a network of multiple nodes via information exchanging in an uncertain environment.
Each node has a linear observation of an unknown parameter with randomly time-varying observation matrices.
The underlying communication network is modeled by a sequence of random digraphs and is subjected to nonuniform random time-varying delays in channels.
Each node runs an online estimation algorithm consisting of a consensus term taking a weighted sum of its own estimate and neighbours' delayed estimates, and an innovation term processing its own new measurement at each time step.  By stochastic time-varying system, martingale convergence theories  and the  binomial expansion of random matrix products,  we transform  the convergence analysis of the algorithm into that of the mathematical expectation of random matrix products. Firstly, for the delay-free case, we show that the algorithm gains can be designed properly such that all nodes' estimates  converge to the true parameter in mean square and almost surely if the observation matrices and communication graphs satisfy the \emph{stochastic spatio-temporal persistence of excitation} condition.
Secondly, for the case with time delays, we introduce \emph{delay matrices} to model the random time-varying communication delays between nodes. It is shown that under the \emph{stochastic spatio-temporal persistence of excitation} condition, for any given bounded delays, proper algorithm gains can be designed to guarantee mean square convergence for the case with conditionally balanced digraphs.
\end{abstract}

\newpage


\section{introduction}
Estimation algorithms have important applications in many fields, e.g. navigation systems, space exploration, machine learning and power systems (\cite{Power system state estimation}-\cite{LO2015SICON}), etc. In a power system, measurement devices such as remote terminal units and phasor measurement units, send the measured active and reactive power flows, bus injection powers and voltage amplitudes to the Supervisory Control and Data Acquisition (SCDA) system, then the voltage amplitudes and phase angles at all buses  are estimated for secure and stable operation of the system (\cite{A distributed state estimation}-\cite{Parallel and distributed}). Generally speaking, there are mainly two categories of estimation algorithms in term of information structure, i.e. centralized and decentralized algorithms. In a centralized algorithm, a fusion center is used to collect all nodes's measurements
and gives the global estimate. This structure heavily relies on the fusion center and lacks robustness and security. In a decentralized algorithm, a network of multiple nodes is employed to cooperatively estimate the unknown parameter via information exchanging, where each node is an entity with integrated capacity of sensing, computing and communication, and occasional node/link failures may not destroy the entire estimation task. Hence, decentralized cooperative estimation algorithms are more robust than centralized ones  (\cite{Schizas1}-\cite{Das}).

There exist various kinds of uncertainties in real networks. For example, sensors are usually powered by chemical or solar cells, and the unpredictability of cell power leads to random node/link failures, which can be modeled by a sequence of random communication graphs. Besides, node sensing failures or measurement losses (\cite{NASSERE}) can be modeled by a sequence of random observation matrices.
There are lots of literature on decentralized online estimation problems with random graphs. Ugrinovskii \cite{Ugrinovskii}  studied  decentralized estimation  with Markovian switching graphs.  Kar \& Moura  \cite{GossipKAR} and  Sahu $et~al$  \cite{CIRFE} considered   decentralized estimation  with i.i.d. graph sequences, where  Kar \& Moura  \cite{GossipKAR}  showed that the algorithm achieves weak consensus under a weak distributed detectability condition and  Sahu $et~al$  \cite{CIRFE} proved that the algorithm  converges almost surely if the mean graph is balanced and strongly connected. Sim\~{o}es \& Xavier \cite{zhongshi} proposed a  decentralized estimation algorithm with i.i.d. undirected graphs and proved that the convergence rate of mean square estimation error is asymptotically equal to that of the centralized algorithm. Decentralized cooperative online estimation based on  diffusion strategies  was addressed in \cite{Lopes}-\cite{Abdolee} with
spatio-temporally independent observation matrices, i.e.  the sequence of observation matrices of each node is an independent random process and  those of different nodes are mutually independent.
Piggott \& Solo \cite{Piggott}-\cite{Piggott1}  studied  decentralized estimation with   temporally correlated observation matrices and a fixed communication graph. Ishihara \& Alghunaim \cite{Ishihara} studied  decentralized estimation with spatially independent observation matrices.
Kar $et~al$~\cite{Soummya1} and Kar \& Moura \cite{SKar} proposed consensus+innovations  decentralized estimation algorithms with random  graphs and observation matrices, where the sequences of communication graphs and observation matrices are both i.i.d. They proved that the algorithm   converges almost surely if the mean graph is balanced and strongly connected.
Zhang \& Zhang \cite{Qiang} considered  decentralized estimation with finite Markovian switching graphs and i.i.d. observation matrices, and proved that the algorithm converges in mean square and almost surely if all graphs are balanced and jointly contain a spanning tree. Zhang $et~al$ \cite{JunfengZhang} proposed a robust  decentralized estimation algorithm with the communication graphs and observation matrices being  mutually independent with each other and both uncorrelated sequences.
In summary, most existing literature on  decentralized cooperative estimation algorithms required balanced mean graphs and special statistical properties of communication graphs and observation matrices, such as i.i.d. or Markovian switching graph sequences,  spatially or temporally independent observation matrices with the fixed mathematical expectation, which are also independent of communication graphs.

Besides random communication graphs and observation matrices, random communication delays are also common in real systems (\cite{Robust Control and Filtering}-\cite{Tian2017TAC}). Due to congestions of communication links and external interferences, time delays are usually random and time-varying, whose  probability distribution can be approximately estimated by statistical methods. However, to our best knowledge, there has been no literature on  decentralized online estimation  with general random time-varying communication delays. Zhang~$et~al$
\cite{zhangya} and Mill\'{a}n~$et~al$ \cite{Mill} considered  decentralized estimation with uniform deterministic time-invariant and time-varying communication delays, respectively, where Mill\'{a}n~$et~al$ \cite{Mill} established  a LMI type convergence condition  by the Lyapunov-Krasovskii functional method.

In this paper, we analyze convergence of  decentralized cooperative online parameter estimation algorithms with random observation matrices, communication graphs and time delays.
Each node's algorithm consists of a consensus term taking a weighted sum of its own estimate and delayed estimates of its neighbouring nodes, and an innovation term processing its own new measurement at each time step.
 The sequences of observation matrices, communication graphs and time delays are not required to satisfy special statistical properties, such as mutual independence and spatio-temporal independence. Furthermore, neither the sample paths of the random graphs nor the mean graphs are necessarily balanced and connected at each time step. These relaxations together with the existence of random time-varying time delays bring essential difficulties to the convergence analysis, and most existing methods are not applicable. For example, the frequency domain approach (\cite{zhangya},\cite{o11po}) is only suitable for deterministic uniform time-invariant time delays, and the Lyapunov-Krasovskii functional method  leads to a non-explicit LMI type convergence condition (\cite{Mill}).   Liu $et~al$ \cite{ouis} and Liu $et~al$  \cite{ouiss}   addressed distributed consensus  with deterministic time-varying communication delays and i.i.d. communication  graphs. The analysis method therein required the  mean graph to be time-invariant and connected at each time step, and is not applicable to time-varying mean graphs.

We introduce \emph{delay matrices} to model the random time-varying communication delays between each pair of nodes. By stochastic time-varying system, martingale convergence theories  and the  binomial expansion of random matrix products,  we transform  the convergence analysis of the algorithm into that of the mathematical expectation of random matrix products. Firstly, for the delay-free case,
we  show that the algorithm gains can be designed properly such that all nodes' estimates  converge to the true parameter in mean square and almost surely if  the observation matrices and communication graphs satisfy the \emph{stochastic spatio-temporal persistence of excitation} conditions.   Especially,  it is shown that for Markovian switching communication graphs and observation matrices,  this condition holds if the stationary graph is balanced with a spanning tree and the measurement model is \emph{spatio-temporally jointly observable}.
Secondly, for the case with time delays, we propose several conditions for  mean square convergence, which explicitly relies on the conditional expectations of delay matrices, observation matrices and weighted adjacency matrices of communication graphs over a sequence of fixed-length time intervals. Furthermore, we show that if the communication graphs are conditionally balanced, then under the \emph{stochastic spatio-temporal persistence of excitation} condition,   for any given bounded delays, proper algorithm gains can be designed to guarantee mean square convergence of the algorithm.
Compared with the existing literature, our contributions are summarized as below.

\begin{itemize}
 \item The delay-free case
      \begin{itemize}
        \item  We show that it is not necessary that the sequences of observation matrices and communication graphs be mutually independent or spatio-temporally independent. Also, the mean graphs are  not necessarily time-invariant and balanced. We establish the \emph{stochastic spatio-temporal persistence of excitation} condition under which the  algorithm with random  graphs  and observation matrices  converges in mean square and almost surely. For a network consisting of completely isolated nodes, the \emph{stochastic spatio-temporal persistence of excitation} condition degenerates to a set of independent \emph{stochastic persistence of excitation} conditions for centralized algorithms (\cite{guolaoshi}).

       \item
        Especially,  for the case with Markovian switching  communication graphs and observation matrices, we prove that the \emph{stochastic spatio-temporal persistence of excitation} condition holds if the stationary graph  is balanced with a spanning tree and the measurement model is spatio-temporally jointly observable, implying that \emph{neither local observability of each node nor instantaneous global observability of the entire measurement model is necessary}.

      \end{itemize}

          \item The case with time delays
          \begin{itemize}
          \item

          We introduce   \emph{delay matrices} to model the random time-varying time delays between each pair of nodes.
          By the method of binomial expansion of random matrix products, we obtain several conditions for mean square convergence,
          which explicitly relies on the conditional expectations of the delay matrices, observation matrices and weighted adjacency matrices of communication graphs   over a sequence of fixed-length time intervals. These conditions show that for given algorithm gains, the communication graphs and observation matrices need to be persistently  excited  with enough intensity  to mitigate the random time delays.
           We further show that if the \emph{stochastic spatio-temporal persistence of excitation} condition holds, then for any given bounded delays, proper algorithm gains can be designed to guarantee mean square convergence of the algorithm for the case with conditionally balanced digraphs.
              \item
                  The nonuniform random time-varying communication delays considered in this paper are more general, and we allow correlated communication delays, graphs and observation matrices.
              \end{itemize}

        \end{itemize}

The rest of the paper is arranged as follows. In Section II, we formulate the problem. In Section III, we describe the decentralized cooperative online  parameter estimation algorithm  with random observation matrices, communication graphs and time delays.   We make the convergence analysis for the delay-free case and the case with time delays in Sections IV and V, respectively. In Section \ref{zhen}, we give a numerical example to demonstrate the theoretical results.
Finally, we conclude the paper and give some future topics in Section VII.

Notation and symbols:
$\circ$: the Hadamard product;
$\otimes$: the Kronecker product;
$\mathrm{Tr}(A)$: the trace of  matrix $A$;
$\|A\|$: the 2-norm of matrix $A$;
$A^T$: the transpose of matrix~$A$;
$\mathbb P\{A\}$: the  probability of  event~$A$;
$I_n$: the $n$ dimensional identity matrix;
$\rho(A)$: the spectral radius of matrix $A$;
$|a|:$ the absolute value of real number~$a$;
$\mathbb R^n$:  the $n$ dimensional real vector space;
$A\ge B$: the matrix $A-B$ is positive semidefinite;
$\lfloor x \rfloor$: the largest integer less than or equal to~$x$;
$\lceil x\rceil$: the smallest integer greater than or equal to~$x$;
$\mathbb E[\xi]$: the mathematical expectation of random variable $\xi$;
$\lambda_{\min}(A)$: the minimum eigenvalue of real symmetric matrix~$A$;
$\textbf{1}_n$:  the $n$ dimensional column vector with all entries being one;
$\textbf{0}_{n\times m}$: the $n\times m$~dimensional matrix with all entries being zero;
$b_n=O(r_n)$: $\lim\sup_{n\to\infty}\frac{|b_n|}{r_n}<\infty$, where~$\{b_n,n\ge0\}$ is a sequence of real numbers,
$\{r_n,n\ge0\}$ is a sequence of real positive numbers;
~$b_n=o(r_n)$: $\lim_{n\to\infty}\frac{b_n}{r_n}=0$;
For a sequence of~$n\times n$ dimensional matrices~$\{Z(k),k\ge0\}$ and a sequence of scalars
~$\{c(k),k\ge0\}$, denote
\bann
\Phi_Z(j,i)= \left\{\begin{array}{cc}
                                Z(j)\cdots Z(i),&~j\ge i \\
                               I_n, &~j< i.
                             \end{array}\right.
\mathrm{and}~~\prod_{k=i}^jc(k)= \left\{\begin{array}{cc}
                                c(j)\cdots c(i),&~j\ge i \\
                               1, &~j< i.
                             \end{array}\right.
\eann
For any nonnegative integers $i$ and $j$, denote the Kronecker function by
$\mathcal I_{i,j}$, satisfying  $\mathcal I_{i,j}=1$
if $i=j$ and $\mathcal I_{i,j}=0$ otherwise.

\section{ problem formulation }\label{tichu}

\subsection{Measurement model}
Consider a network  of $N$ nodes. Each node is an estimator with integrated capacity of sensing, computing, storage and communication. The estimators/nodes cooperatively estimate an unknown parameter vector $x_0\in\mathbb R^n$ via information exchanging. The relation between the measurement vector $z_i(k)\in\mathbb R^{n_i}$  of estimator $i$ and the unknown parameter $x_0$ is represented by
\bna\label{op09}
z_i(k)=H_i(k)x_0+v_i(k),\ i=1,\cdots,N,\ k\ge0.
\ena
Here, $H_i(k)\in\mathbb R^{n_i\times n}$ is the random observation (regression) matrix  at time instant $k$ with $n_i\le n$, and $v_i(k)\in\mathbb R^{n_i}$ is the additive measurement noise.
Denote $z(k)=[z_1^T(k),\cdots,z_N^T(k)]^T$, $H(k)=[H_1^T(k),\cdots,H_N^T(k)]^T$ and $v(k)=[v_1^T(k),\cdots,v_N^T(k)]^T$.
Rewrite (\ref{op09}) by  the compact form
\bna\label{98isasao}
z(k)=H(k)x_0+v(k),\ k\ge0.
\ena

\vskip 0.2cm

\begin{remark}
\rm{In many real applicaitons, the relations between the unknown parameter and the measurements can be represented by (\ref{op09}). For example,  in the  decentralized multi-area state estimation in power systems, the grid is partitioned into multiple geographically non-overlapping areas, and each area is regarded as a node. The  grid state $x_0$  to be estimated  consists of voltage amplitudes and phase angles at all buses. The measurement $z_i(k)$ of each area/node consists of the active and reactive power flow, bus injection powers and voltage amplitude information measured by remote terminal units and phasor measurement units in the $i$-th area. By the DC power
flow approximation (\cite{lll}), the grid state degenerates to the voltage phase angles at all buses and the relation between the measurement of each area and the grid state can be represented by (\ref{op09}).
In  decentralized parameter identification, each node's measurement equation is given by
\bann
z_i(k)=\sum_{j=1}^nc_jz_i(k-j)+v_i(k)=[z_i(k-1),\cdots, z_i(k-n)][c_1,\cdots,c_n]^T+v_i(k).
\eann
For this case, the  unknown parameter $x_0=[c_1,\cdots,c_n]^T$ and the observation matrix (generally called   regressor) $H_i(k)=[z_i(k-1),\cdots, z_i(k-n)]$ is an $n$ dimensional row vector. In addition, sensing failures in real networks  can be modeled by a Markov chain or an i.i.d. sequence of Bernoulli variables $\{\delta_i(k),k\ge0\}$. Then $H_i(k)=\delta_i(k)H_i'(k)$, where $\{H_i'(k),k\ge0\}$ is the sequence of observation matrices without sensing failures.}
\end{remark}

\subsection{Communication models}\label{delaymodel}
 Assume that there exist nonuniform random time-varying communication delays for the communication links between each pair of nodes. We use a sequence of random variables $\{\lambda_{ji}(k)\in\{0,\cdots,d\}$, $k\ge0\}$ to represent the time delays associated with the link from node $j$ to node~$i$, where the positive integer~$d$ represents the maximum time delay. This sequence is subjected to the  discrete probability distribution
\bna\label{098}
\mathbb P\{\lambda_{ji}(k)=q\}=p_{ji,q}(k)~\mathrm{with}~\sum_{q=0}^dp_{ji,q}(k)=1.
\ena
We stipulate that $\mathbb P\{\lambda_{ii}(k)=0\}=1$, $i=1,\cdots, N$, $k\ge0$. Denote the $N$ dimensional matrices $\mathcal I(k,q)=[\mathcal I_{\lambda_{ji}(k),q}]_{1\le j,i\le N}$, $0\le q\le d$, $k\ge0,$ called   \emph{delay matrices}. By the definition of Kronecker function, we know that for each~$q=0,1,...,d$, $\{\mathcal I(k,q),k\ge0\}$ is a sequence of random matrices and its sample paths are sequences of~$0-1$ matrices. By (\ref{098}), we know that~$\mathbb E[\mathcal I_{\lambda_{ji}(k),q}]=p_{ji,q}(k)$ and
\bna\label{554rf9s}
\sum_{q=0}^d\mathcal I(k,q)=\textbf{1}_N\textbf{1}_N^T~\rm{\rm{a.s.}}
\ena

We use a sequence of random communication graphs $\{\mathcal G(k)=\langle\mathcal V$, $\mathcal A_{\mathcal G(k)}\rangle$, $k\ge0\}$ to describe the possible link failures among nodes, where $\mathcal V=\{1,\cdots,N\}$ is the node set and $\mathcal A_{\mathcal G(k)}=[a_{ij}(k)]_{1\le i,j\le N}$ is the  weighted adjacency matrix of the communication graph, in which $a_{ii}(k)=0$ a.s. for all $i\in\mathcal V$ and $k\ge0$ and $a_{ij}(k)\not=0$ if and only if the link from node $j$ to node $i$ exists at time instant $k$ for all $i\not=j$.
The neighborhood  of node $i$ is~$\mathcal N_i(k)=\{j|a_{ij}(k)\not=0\}$. The  degree matrix of the graph is~$\mathcal D_{\mathcal G(k)}=diag(\sum_{j=1}^Na_{1j}(k),\cdots,\sum_{j=1}^Na_{Nj}(k))$ and the Laplacian matrix of the graph is~$\mathcal L_{\mathcal G(k)}=\mathcal D_{\mathcal G(k)}-\mathcal A_{\mathcal G(k)}$ (\cite{iwerw}\cite{tau}). Denote $\widehat {\mathcal L}_{\mathcal G(k)}=\frac{\mathcal L_{\mathcal G(k)}+\mathcal L^T_{\mathcal G(k)}}{2}$.
Specifically, if~${\mathcal G(k)}$ is balanced, then $\widehat {\mathcal L}_{\mathcal G(k)}$ is the Laplacian matrix of the symmetrized graph of~${\mathcal G(k)}$, $k\ge0$ (\cite{tau}).
Let
\bnaa\label{akdef}
\overline A(k,q)=(\mathcal A_{\mathcal G(k)}\circ \mathcal I(k,q))\otimes I_n.
\enaa
Then, by~(\ref{554rf9s}) and the above, we have
\bnaa\label{0fkwjeewa}
\sum_{q=0}^d\overline A(k,q)=\mathcal A_{\mathcal G(k)}\otimes I_n.
\enaa

\section{ decentralized cooperative online estimation algorithm}\label{algorithm}

{Let~$ x_i(k)\in\mathbb R^n$~be the estimate by node $i$ for the unknown parameter $x_0$ at time  instant $k,k\ge-d$ with the initial estimates $x_i(k),-d\le k\le 0$ being any given real vectors.} Starting at the initial estimate, at any time instant~$k\ge0$, node~$i$ takes a weighted sum of its own estimate and delayed estimates received from its neighbours, and then adds a correction term based on the local measurement information (innovation) to update  the estimate~$x_i(k+1)$. Specifically, the  decentralized cooperative online parameter estimation algorithm with random observation matrices, communication graphs and time delays, motivated by a baseline version without time delays in \cite{SKar},  is given by
\bna\label{asapp}
x_i(k+1)&=&x_i(k)+a(k)H_i^T(k)(z_i(k)-H_i(k) x_i(k))\cr
&&+b(k)\sum_{j\in\mathcal  N_i(k)}a_{ij}(k)( x_j(k-\lambda_{ji}(k))- x_i(k)),\
  i\in\mathcal V,\ k\ge0,
\ena
where $a(k)$ and $b(k)$ are the  innovation  and  consensus algorithm gains, respectively.

\rm{Denote the $\sigma-$fileds $\mathcal F(k)=\sigma( \mathcal A_{\mathcal G(s)},v(s),H_i(s),\lambda_{ji}(s),~j,i\in\mathcal{V},~0\le s\le k)$, $k\ge0$, with $\mathcal F(-1)=\{\Omega,\emptyset\}$.
 For the algorithm (\ref{asapp}), we have the following assumptions.

\textbf{A1.a} The sequence $\{v(k),k\ge0\}$ is independent of $\{H(k)$, $k\ge0\}$, $\{\mathcal A_{\mathcal G(k)}$, $k\ge0\}$ and $\{\lambda_{ji}(k)$, $j, i\in\mathcal{V}$, $k\ge0\}$.

\textbf{A1.b} The  sequence $\{v(k),\mathcal F(k),k\ge0\}$ is a martingale difference sequence and there exists a constant~$\beta_v>0$ such that~$\sup_{k\ge0}\mathbb E[\|v(k)\|^2|\mathcal F(k-1)]\le\beta_v~\rm{a.s.}$

\textbf{A2.a} $\sup_{k\ge0}\|H(k)\|<\infty~\mathrm{a.s.}\ \mbox{and}\ \sup_{k\ge0}\|\mathcal A_{\mathcal G(k)}\|<\infty~\rm{a.s.}$

\textbf{A2.b}  \rm{There exist positive constants~$\beta_a$ and $\beta_H$ such that~$\max_{i,j\in\mathcal V}\sup_{k\ge0}| a_{ij}(k)|\le\beta_a~\rm{a.s.}$ and $\max_{i\in\mathcal V}\sup_{k\ge0}\| H_i(k)\|\le\beta_H~\mathrm{a.s.}$}

For the algorithm gains, we make the following conditions.

\textbf{C1.a} The sequences \rm{$\{a(k),k\ge0\}$ and $\{b(k),k\ge0\}$ are positive real sequences  monotonically decreasing to zero,
satisfying  $a(k)=O(b(k))$. }

\textbf{C1.b}  $b^2(k)=o(a(k)),a(k)=O(a(k+1))$ and $\sum_{k=0}^\infty a(k)=\infty.$

\textbf{C1.c} \rm{$\sum_{k=0}^\infty b^2(k)<\infty$.}

\vskip 0.2cm

\begin{remark}
\rm{Note that, in Assumption \textbf{A1.a}, neither mutual independence nor spatio-temporal  independence is assumed on the observation matrices, communication graphs and time delays. }
\end{remark}

\vskip 0.2cm
\begin{remark}
\rm{ It is easy to find $\{a(k)$, $k\ge0\}$ and $\{b(k)$, $k\ge0\}$  satisfying Conditions \textbf{C1.a}--\textbf{C1.c}. If  $a(k)=\frac{1}{(k+1)^{\tau_1}}$, $b(k)=\frac{1}{(k+1)^{\tau_2}}$, $k\geq0$, $0.5<\tau_2\leq\tau_1\leq1$, then these conditions hold.  }
\end{remark}

By the definition of $\mathcal I_{\lambda_{ji}(k),q}$, we know that $x_j(k-\lambda_{ji}(k))=\sum_{q=0}^d x_j(k-q)\mathcal I_{\lambda_{ji}(k),q}$. Then by~(\ref{asapp}), we have
\bna\label{asapp1}
x_i(k+1)&=&x_i(k)+a(k)H_i^T(k)[z_i(k)-H_i(k) x_i(k)]\cr
&&+b(k)\sum_{j\in\mathcal  N_i(k)}a_{ij}(k)\Bigg[\sum_{q=0}^d x_j(k-q)\mathcal I_{\lambda_{ji}(k),q}- x_i(k)\Bigg],\ i\in\mathcal V.
\ena
Denote ${{\mathcal H}(k)}=diag\{H_1(k),\cdots,H_N(k)\}$ and $x(k)=[x_1^T(k),\cdots,x_N^T(k)]^T$. By (\ref{akdef}), rewrite (\ref{asapp1}) as
\bna\label{s9s}
x(k+1)&=&[I_{Nn}-b(k)\mathcal D_{\mathcal G(k)}\otimes I_n-a(k){\mathcal H}^T(k){\mathcal H}(k)]x(k)\cr
&&+b(k)\sum_{q=0}^d\overline A(k,q)x(k-q)+a(k){\mathcal H}^T(k)z(k).
\ena
Denote the overall estimation error vector~$e(k)=x(k)-\textbf{1}_N\otimes x_0$. Note that~$(\mathcal L_{\mathcal G(k)}\otimes I_n)(\textbf{1}_N\otimes x_0)=0$. By (\ref{98isasao}) and (\ref{0fkwjeewa}), subtracting $\textbf{1}_N\otimes x_0$ on both sides of~(\ref{s9s}) leads to
\ban\label{asd2dsa}
&&e(k+1)\cr
&=&[I_{Nn}-b(k)\mathcal D_{\mathcal G(k)}\otimes I_n-a(k){\mathcal H}^T(k){\mathcal H}(k)]x(k)+b(k)\sum_{q=0}^d\overline A(k,q)x(k-q)\cr
&&+a(k){\mathcal H}^T(k)z(k)-\textbf{1}_N\otimes x_0\cr
&=&[I_{Nn}-b(k)\mathcal D_{\mathcal G(k)}\otimes I_n-a(k){\mathcal H}^T(k){\mathcal H}(k)](x(k)-\textbf{1}_N\otimes x_0+\textbf{1}_N\otimes x_0)\cr
&&+b(k)\sum_{q=0}^d\overline A(k,q)(x(k-q)-\textbf{1}_N\otimes x_0+\textbf{1}_N\otimes x_0)+a(k){\mathcal H}^T(k)z(k)-\textbf{1}_N\otimes x_0\cr
&=&[I_{Nn}-b(k)\mathcal D_{\mathcal G(k)}\otimes I_n-a(k){\mathcal H}^T(k){\mathcal H}(k)](e(k)+\textbf{1}_N\otimes x_0)\cr
&&+b(k)\sum_{q=0}^d\overline A(k,q)(e(k-q)+\textbf{1}_N\otimes x_0)+a(k){\mathcal H}^T(k)z(k)-\textbf{1}_N\otimes x_0\cr
&=&[I_{Nn}-b(k)\mathcal D_{\mathcal G(k)}\otimes I_n-a(k){\mathcal H}^T(k){\mathcal H}(k)]e(k)-a(k){\mathcal H}^T(k){\mathcal H}(k)(\textbf{1}_N\otimes x_0)\cr
&&+b(k)\sum_{q=0}^d\overline A(k,q)e(k-q)+a(k){\mathcal H}^T(k)H(k)x_0+a(k){\mathcal H}^T(k)v(k),
\ean
which together with ${\mathcal H}(k)(\textbf{1}_N\otimes x_0)= H(k)x_0$ gives the overall estimation error equation
\bnaa\label{s1119s}
e(k+1)&=&[I_{Nn}-b(k)\mathcal D_{\mathcal G(k)}\otimes I_n-a(k){\mathcal H}^T(k){\mathcal H}(k)]e(k)\cr
&& +b(k)\sum_{q=0}^d\overline A(k,q)e(k-q)+a(k){\mathcal H}^T(k)v(k),~k\ge0.
\enaa
For the delay-free case,  $d=0$. Then the algorithm (\ref{s9s}) becomes
\bna\label{s9sdelayfree}
x(k+1)=[I_{Nn}-b(k)\mathcal L_{\mathcal G(k)}\otimes I_n-a(k){\mathcal H}^T(k){\mathcal H}(k)]x(k)+a(k){\mathcal H}^T(k)z(k),
\ena
and  the estimation error equation (\ref{s1119s}) becomes
\bnaa
\label{zz9qjasoda}
e(k+1)=[I_{Nn}-b(k)\mathcal L_{\mathcal G(k)}\otimes I_n-a(k){\mathcal H}^T(k){\mathcal H}(k)]e(k)+a(k){\mathcal H}^T(k)v(k).
\enaa

\vskip 0.2cm
\begin{remark}
\rm{In this paper, we use the concept of Laplacians  of digraphs defined in \cite{tau}, which is widely used in the literature on decentralized estimation (\cite{GossipKAR}--\cite{Qiang}). For the delay-free case, the consensus term $[\mathcal L_{\mathcal G(k)}\otimes I_n]x(k)$ naturally appears in the algorithm (\ref{s9sdelayfree}). Note that another concept of symmetric Laplacians of digraphs is proposed in \cite{Chung}. This symmetric Laplacian involves the Perron vector of the weighted adjacency matrix. It has been pointed out  in \cite{Chung} that for a general digraph, there is no closed
form solution for the Perron vector. Generally, the $i$th element of the Perron vector,  which is not local information of the $i$th node, depends on the weights associated to all nodes. Therefore, though the Laplacian
proposed in \cite{Chung} is symmetric, it is generally incompatible with the decentralized nature of the estimation algorithm.}
\end{remark}

\section{ the delay-free case}\label{jieguo}
In this section, we give the convergence conditions of the  algorithm (\ref{asapp}) for the delay-free  case, i.e. $\lambda_{ji}(k)=0$, a.s. $\forall\ j, i\in\mathcal V$, $\forall\ k\ge0$. All proofs of  results are put in Appendix B.

For any given positive integers $h$ and $m$, denote
\bann
{\Lambda_m^h}&=&\lambda_{\min}\Bigg[\sum_{k=mh}^{(m+1)h-1}\Big(\mathbb E[\widehat {\mathcal L}_{\mathcal G(k)}|\mathcal F(mh-1)]\otimes I_n+\mathbb E[{\mathcal H}^T(k){\mathcal H}(k)|\mathcal F(mh-1)]\Big)\Bigg],\cr
\overline{\Lambda}_m^h&=&\lambda_{\min}\Bigg[\sum_{k=mh}^{(m+1)h-1}\Big(b(k)\mathbb E[\widehat {\mathcal L}_{\mathcal G(k)}|\mathcal F(mh-1)]\otimes I_n+a(k)\mathbb E[{\mathcal H}^T(k){\mathcal H}(k)|\mathcal F(mh-1)]\Big)\Bigg].
\eann

We first give a result for the case with general processes of random graphs and observation matrices.

\begin{theorem}\label{nodelay111}
\rm{Suppose that Assumptions~\textbf{A1.a}--\textbf{A1.b} hold. If Condition   \textbf{C1.a} holds, and there exists an integer $h>0$,  a constant $\rho_0>0$ and a positive real sequence $\{c(m),m\ge 0\}$  with
\bnaa\label{cmsequence}
b^2(mh)=o(c(m)),\ \sum_{m=0}^{\infty}c(m)=\infty,
\enaa
such that (b.1) $\overline{\Lambda}_m^h\ge c(m)$\ \rm{a.s.}, $m\geq0$
and (b.2) $\sup_{k\ge0}[\mathbb E[(\|\mathcal L_{\mathcal G(k)}\|+\|{\mathcal H}^T(k){\mathcal H}(k)\|)^{2^{\max\{h,2\}}}$ $|\mathcal F(k-1)]]^{\frac{1}{2^{\max\{h,2\}}}}\le\rho_0~\rm{a.s.},
$
then the algorithm~(\ref{asapp})  converges in mean square, that is, $\lim_{k\to\infty}$$\mathbb E$$\|x_i(k)$ $-x_0\|^2=0,~i\in\mathcal V.$ In addition, if Assumption~\textbf{A2.a} and Condition \textbf{C1.c} hold, then the  algorithm~(\ref{asapp}) converges almost surely, i.e.~$\lim_{k\to\infty}x_i(k)=x_0,~i\in\mathcal V~\rm{a.s.}$ }
\end{theorem}

\vskip 0.2cm
\begin{remark}\label{sajdqwokq0apspasd}
\rm{Most existing literature on    decentralized estimation suppose that the mean graphs are  balanced (\cite{Soummya1},\cite{Qiang}).  Here, the condition (b.1) in Theorem \ref{nodelay111} may still hold even if the mean graphs are unbalanced.  For example,  consider a  fixed  weighted graph $\mathcal G=\langle\mathcal V=\{1,2\}, \mathcal A_{\mathcal G}=[a_{ij}]_{2\times2}\rangle$ with $a_{12}=1$ and $ a_{21}=0.3$. Obviously, $\mathcal G$ is unbalanced.  Suppose $H_1=0,H_2=1$. Choose $a(k)=b(k)=\frac{1}{k+1}$.
We have $\lambda_{\min}(b(m)\widehat{\mathcal L}_{\mathcal G}+a(m){\mathcal H}^T\mathcal H)=\frac{1}{m+1}\lambda_{\min}(\widehat{\mathcal L}_{\mathcal G}+{\mathcal H}^T\mathcal H)=\frac{0.5821}{m+1}$.
Then,   the condition   (b.1) holds with $h=1$ and $c(m)=\frac{0.5821}{m+1}$ satisfying (\ref{cmsequence}). A more complex example with unbalanced mean graphs is given in Section \ref{zhen}.  }
\end{remark}

\vskip 0.2cm

Next, we give Theorem \ref{nodelay} for the case with   conditionally balanced digraphs:
\ban
\Gamma_1&=&\Big\{\{\mathcal G(k),k\ge0\}|\mbox{the random matrix}~\mathbb E[\mathcal A_{\mathcal G(k)}|\mathcal F(k-1)]\mbox{~is nonnegative}\cr
&&~~~~~~~~~~~~~~~~~~~~\mbox{and its associated random graph is balanced a.s.},k\ge0\Big\}.
\ean

\begin{theorem}\label{nodelay}
\rm{Suppose that $\{\mathcal G(k), k\ge0\}\in\Gamma_1$, Assumptions~\textbf{A1.a}--\textbf{A1.b} hold. If Conditions \textbf{C1.a}--\textbf{C1.b} hold,   and there exists an integer $h>0$, positive constants $\theta$ and $\rho_0$ such that (c.1)
$\inf_{m\ge 0}{\Lambda_m^h}\ge\theta>0~\rm{a.s.}$  and
(c.2) $\sup_{k\ge0}[\mathbb E[(\|\mathcal L_{\mathcal G(k)}\|+\|{\mathcal H}^T(k){\mathcal H}(k)\|)^{2^{\max\{h,2\}}}|\mathcal F(k-1)]]^{\frac{1}{2^{\max\{h,2\}}}}\le\rho_0~\rm{a.s.}$,
then the algorithm~(\ref{asapp})  converges in mean square. In addition, if Assumption~\textbf{A2.a} and Condition \textbf{C1.c} hold, then the  algorithm~(\ref{asapp}) converges almost surely.}
\end{theorem}

\vskip 0.2cm

\begin{remark}
\rm{The condition (b.1) in Theorem \ref{nodelay111} and the condition (c.1) in Theorem \ref{nodelay} are the key convergence conditions. We call them the \emph{stochastic spatio-temporal persistence of excitation} conditions.
In detail, \textit{spatio} emphasizes the reliance of the conditions on the  communication graphs and observation matrices over all nodes rather than a single node, while \textit{temporal} represents the summing matrices over a sequence of fixed-length time intervals rather than a single time step, and ``persistence of excitation''  represents that the minimum eigenvalues of matrices consisting of spatio-temporal  observation matrices and Laplacian matrices  are uniformly bounded away from zero with respect to the sample paths in some sense. Guo \cite{guolaoshi} considered centralized estimation algorithms with random observation matrices and proposed the  ``stochastic persistence of excitation'' condition to ensure convergence.
The condition (c.1) can be regarded as the generalization of ``stochastic persistence of excitation'' condition in \cite{guolaoshi} to that for  decentralized algorithms. For a  network with $N$  isolated nodes, $\mathcal L_{\mathcal G(k)}\equiv\textbf{0}_{N\times N}$ a.s., and the condition (c.1) degenerates to $N$ independent ``stochastic persistence of excitation'' conditions.}
\end{remark}

\vskip 0.2cm

In the most existing literature, it was also required that the sequence of observation matrices
be i.i.d. and independent of the sequence of communication graphs, neither of which is necessary
in Theorems \ref{nodelay111} and \ref{nodelay}.
Subsequently, we give more intuitive convergence conditions
for  the case with Markovian switching communication graphs and observation matrices. We first make the following assumption.

\vskip 0.2cm

\textbf{A3} $\{ \langle\mathcal H(k),\mathcal A_{\mathcal G(k)} \rangle, k\ge0\}\subseteq {\mathcal S}$ is a homogeneous and uniform ergodic Markov chain  with  a unique stationary distribution $ \pi$.\\
Here,  $ {\mathcal S}=\{ \langle\mathcal H_l, \mathcal A_l \rangle,l=1,2,...\}$ with $\mathcal H_l=diag (H_{1,l},\cdots, H_{N,l})$, where $\{H_{i,l}\in\mathbb R^{n_i\times n},l=1,2,...\}$ is the state space of  observation matrices of node $i$ and $\{\mathcal A_l,l=1,2,...\}$ being the state space of the weighted adjacency matrices, $\pi=[\pi_1, \pi_2,...]^T$, $\pi_{l}\geq0$, $l=1,2,...$, and $\sum_{l=1}^{\infty}\pi_{l}=1$ with  $\pi_l$ representing $\pi( \langle\mathcal H_l, \mathcal A_l \rangle)$.

\vskip 0.2cm

\begin{corollary}\label{marke}
\rm{Suppose that Assumptions~\textbf{A1.a}--\textbf{A1.b}, \textbf{A3} hold, and
  $\sup_{l\ge1}\|\mathcal A_l\|<\infty$, $\sup_{l\ge1}\|\mathcal H_l\|<\infty$. If Conditions \textbf{C1.a}--\textbf{C1.c} hold, and

      (d.1)    the stationary weighted adjacency matrix $\sum_{l=1}^\infty\pi_l\mathcal A_l$  is nonnegative and its  associated graph is balanced with a spanning tree;

      (d.2)  the measurement model (\ref{op09}) is \emph{spatio-temporally jointly observable}, i.e.
      \bnaa\label{osnsdfka}
      \lambda_{\min}\Bigg(\sum_{i=1}^N\Bigg(\sum_{l=1}^\infty\pi_lH_{i,l}^TH_{i,l}\Bigg)\Bigg)>0,\enaa
      then the algorithm~(\ref{asapp})  converges  in mean square and almost surely.}
\end{corollary}

\vskip 0.2cm

\begin{remark}
\rm{Most of the existing  decentralized estimation algorithms used the mathematical expectation of observation matrices which is restricted to be time-invariant and difficult to be obtained (\cite{Soummya1},\cite{Qiang}). They required instantaneous global observability in the statistical sense for the measurement model, i.e., $\sum_{i=1}^N\overline H_i^T\overline H_i$ is positive definite, where $\overline H_i$ is a fixed matrix with $\mathbb E[H_i(k)]\equiv \overline H_i$, for all  $k\ge0$, $i=1,2,...,N$.
In contrast, we only use the sample paths of observation matrices in the algorithm (\ref{asapp}). The mathematical expectations of observation matrices are allowed to be time-varying.
 We prove that for  homogeneous and uniform ergodic Markovian switching observation matrices and communication graphs, the \emph{stochastic spatio-temporal persistence of excitation} condition given in Theorem \ref{nodelay} holds if the stationary graph is balanced with a spanning tree and the measurement model is spatio-temporally jointly observable, that is, (\ref{osnsdfka}) holds, implying that neither local observability of each node, i.e. $\lambda_{\min}(\sum_{l=1}^\infty\pi_lH_{i,l}^TH_{i,l})>0$, $i\in\mathcal V$, nor instantaneous global observability  of  the entire measurement model,  i.e. $\lambda_{\min} (\sum_{i=1}^N H_{i,l}^TH_{i,l})>0$, $l=1,2,...$, is needed.}
 \end{remark}

\section{the case with random Time-varying communication delays}\label{delayexist}

In this section, we  analyze the convergence of the algorithm (\ref{asapp}) with random observation matrices, communication graphs and time delays simultaneously. All proofs of results are put in Appendix C.

In the presence of random time-varying communication delays, the mean square convergence analysis of the algorithm becomes very difficult. To address this, we transform the estimation error equation (\ref{s1119s}) into the following equivalent system (\cite{ouis}-\cite{ouiss}).
\bna\label{12129s}
r(k+1)&=&F(k)r(k)+g(k),\cr
g(k)&=&\sum_{q=1}^dC_q(k)g(k-q)+a(k){\mathcal H}^T(k)v(k),\ k\geq0,
\ena
where $F(k)$, $C_q(k)$,  $1\le q\le d$, $k\ge0$ satisfy
\bna\label{121esd29s}
&&F(k)+C_1(k)=I_{Nn}-b(k)\mathcal D_{\mathcal G(k)}\otimes I_n-a(k){\mathcal H}^T(k){\mathcal H}(k)+b(k)\overline A(k,0),\cr
&&C_1(k)F(k-1)-C_2(k)=-b(k)\overline A(k,1),\cr
&&C_2(k)F(k-2)-C_3(k)=-b(k)\overline A(k,2),\cr
&&~~~~~~~~~~~~~~~\vdots\cr
&&C_{d-1}(k)F(k-d+1)-C_d(k)=-b(k)\overline A(k,d-1),\cr
&&C_{d}(k)F(k-d)=-b(k)\overline A(k,d).
\ena
Here, $F(k)=I_{Nn}$, $-d\le k\le-1$. It can be verified that if $r(k)=e(k)$, $-d\le k\le-1$, then $r(k)=e(k)$, $\forall\ k\ge 0$,
i.e. the system (\ref{s1119s}) and the system (\ref{12129s})-(\ref{121esd29s}) are equivalent.

We need the following condition on the consensus gain.

\textbf{C1.d} The initial consensus gain \rm{$b(0)\le \max_{0<\psi<1}f_{C_1,\beta_a,\beta_H,N,d}(\psi)$, where
\bann\label{sad0qiddmpa}
f_{C_1,\beta_a,\beta_H,N,d}(\psi)\triangleq\frac{\psi}{N\beta_a+C_1\beta_H^2+
N\beta_a[(1-\psi)^{-(d+1)}-1]/[(1-\psi)^{-1}-1]},d\ge1,\psi\in(0,1),
\eann
with $C_1\triangleq\sup_{k\ge0}\frac{a(k)}{b(k)}$.}

It  can be verified that given Assumption~\textbf{A2.b} and Condition \textbf{C1.a}, $\max_{0<\psi<1}f_{C_1,\beta_a,\beta_H,N,d}(\psi)$ is well-defined.  Examples of  $f_{1,1,1,N,d}(\cdot)$ with different $d$ and $N$ are shown in Figure \ref{Assumption3cremark}.

We first establish a  lemma as the basis of convergence analysis.
\begin{lemma}\label{asa0}
\rm{If Assumption~\textbf{A2.b}, Conditions \textbf{C1.a} and  \textbf{C1.d} hold,  then
$F(k)$ is invertible  and
$\|F^{-1}(k)\|\le(1-\psi_1)^{-1}$~\rm{a.s.},~$\forall~k\ge 0$, where
$
\psi_1=\min\{\psi\in(0,1)| f_{C_1,\beta_a,\beta_H,N,d}(\psi)\geq b(0)\}.
$
}
\end{lemma}
Note that, by the continuity  of $ f_{C_1,\beta_a,\beta_H,N,d}(\cdot)$ and Conditon \textbf{C1.d}, it is known that the set $\{\psi\in(0,1)| f_{C_1,\beta_a,\beta_H,N,d}(\psi)\geq b(0)\}$ is a nonempty and bounded closed set. Then, $\psi_1$  is well-defined.

\begin{figure}[H]
  \centering
  \includegraphics[width=8cm]{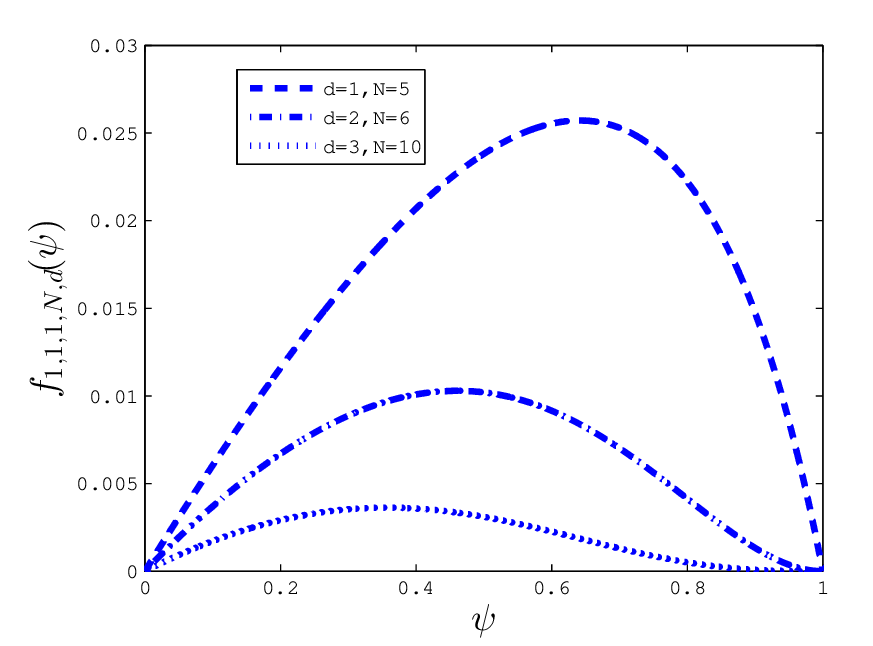}\\
\setlength{\abovecaptionskip}{0pt}
\setlength{\belowcaptionskip}{0pt}
  \caption{\footnotesize{The curves of $f_{1,1,1,N,d}(\cdot)$.}}\label{Assumption3cremark}
\end{figure}
\vspace{-0.8cm}

If the conditions of  Lemma~\ref{asa0} hold, then $F(k)$ is invertible \rm{a.s.} Thus, by (\ref{121esd29s}), we have
\bna\label{ede92}
F(k)&=&I_{Nn}-b(k){\mathcal D}_{\mathcal G(k)}\otimes I_n-a(k){\mathcal H}^T(k){\mathcal H}(k)+b(k)\overline A(k,0)-C_1(k)\cr
&=&I_{Nn}-b(k){\mathcal D}_{\mathcal G(k)}\otimes I_n-a(k){\mathcal H}^T(k){\mathcal H}(k)+b(k)\overline A(k,0)\cr
&&~~~-(C_2(k)-b(k)\overline A(k,1))F^{-1}(k-1)\cr
&&~~~~~~~~~~~~~~~\vdots\cr
&=&I_{Nn}-G(k),~k\ge 0,
\ena
where
\bnaa\label{gkfor}
G(k)\triangleq b(k)\mathcal {\mathcal D}_{\mathcal G(k)}\otimes I_n+a(k){\mathcal H}^T(k){\mathcal H}(k)-b(k)\sum_{q=0}^d\overline A(k,q)\big[\Phi_F(k-1,k-q)\big]^{-1}.
\enaa
For any given positive integers~$h$ and $m$, denote
\bnaa\label{lam}
&&{\widetilde\Lambda_m^h}=\lambda_{\min}\Bigg[\sum_{k=mh}^{(m+1)h-1}\Bigg({b(k)}\mathbb E[\widehat{\mathcal L}_{\mathcal G(k)}|\mathcal F(mh-1)]\otimes I_n+a(k)\mathbb E[{\mathcal H}^T(k){\mathcal H}(k)|\mathcal F(mh-1)]\cr
&&~~~~~~~~~~~~~-\frac{b(k)}{2}\sum_{q=0}^d
\mathbb E\Big[\overline A(k,q)[[\Phi_F(k-1,k-q)]^{-1}-I_{Nn}]\cr
&&~~~~~~~~~~~~~+[[\Phi_F(k-1,k-q)]^{-1}-I_{Nn}]^T\overline A^T(k,q)\Big|\mathcal F(mh-1)\Big]\Bigg)\Bigg].
\enaa

\begin{theorem}\label{e4t5634332}
\rm{Suppose that Assumptions~\textbf{A1.a}--\textbf{A1.b}, \textbf{A2.b} hold. If Conditions \textbf{C1.a}, \textbf{C1.d} hold, and there exists an integer $h>0$ and a positive real sequence $\{ c(m),m\ge0\}$ with $b^2(mh)=o(c(m))$,\ $\sum_{m=0}^{\infty}c(m)=\infty$, such that
\bnaa
\label{saijdq98e00sad}
{\widetilde\Lambda_m^h}\ge  c(m)\ \mathrm{a.s.},\ m\geq0,
\enaa
then the   algorithm~(\ref{asapp})   converges in mean square.}
\end{theorem}


If  $\{ \langle\mathcal H(k),\mathcal A_{\mathcal G(k)},\lambda_{ji}(k),j,i\in\mathcal V \rangle, k\ge0\}$ is an independent random process, then Corollary  \ref{tuiasdwqrqacasdwl} below gives a sufficient condition for the condition (\ref{saijdq98e00sad}) in Theorem \ref{e4t5634332} to  hold, which is more intuitive and computable.

\begin{corollary}\label{tuiasdwqrqacasdwl}
\rm{Suppose that Assumptions~\textbf{A1.a}--\textbf{A1.b}, \textbf{A2.b} hold, $\{ \langle\mathcal H(k),\mathcal A_{\mathcal G(k)},\lambda_{ji}(k),j,i\in\mathcal V \rangle, k\ge0\}$ is an independent process. If   Condition \textbf{C1.a} holds, $b(0)\le f_{C_1,\beta_a,\beta_H,N,d}(\psi_2)$ with $ \psi_2\in(0,{2}^{\frac{1}{d}}-1)$,
and there exists an integer $h>0$ and a positive real sequence~$\{ c(m),m\ge0\}$ with $b^2(mh)=o(c(m))$ and $\sum_{m=0}^{\infty}c(m)=\infty$, such that
\bnaa\label{926hao123}
\overline{\Lambda}_m^h-\sum_{k=mh}^{(m+1)h-1}\Bigg[{b(k)}\sum_{q=0}^d
\frac{\|\mathbb E[\overline A(k,q)]\|[(1+\psi_2)^q-1]}{2-(1+\psi_2)^q}\Bigg]\ge c(m),\ m\geq0,
\enaa
then  the   algorithm~(\ref{asapp})   converges in mean square.}
\end{corollary}

\vskip 0.2cm

  Next,  for the case with conditionally balanced digraphs,  the following corollary presents a more intuitive convergence condition.

\begin{corollary}\label{tuilunasdasaw}
\rm{ Suppose that Assumptions~\textbf{A1.a}--\textbf{A1.b}, \textbf{A2.b} hold and $\{\mathcal G(k), k\ge0\}\in\Gamma_1$.    If   Conditions \textbf{C1.a}--\textbf{C1.b}, \textbf{C1.d}  hold, $b(k)=O(a(k))$,
 and there exists an integer $h>0$, a constant $\theta>0$ such that
  \bnaa\label{ssaoidjaoi90e0}
  \inf_{m\ge0}(\Lambda_m^h-\Sigma_m^h)\ge\theta\ \mathrm{a.s.}
   \enaa
  where $$\Sigma_m^h=C_2 (C_3)^h\max\{1,C_1\}\sum_{k=mh}^{(m+1)h-1}\sum_{q=0}^d\|\mathbb E[\overline A(k,q)([\Phi_F(k-1,k-q)]^{-1}-I_{Nn})|\mathcal F(mh-1)]\|$$ with $C_2\triangleq\sup_{k\ge0}\frac{b(k)}{a(k)}$ and $C_3\triangleq\sup_{k\ge0}\frac{a(k)}{a(k+1)}$,
 then, the  algorithm~(\ref{asapp}) converges in mean square.  Furthermore,   if $\{ \langle\mathcal H(k),\mathcal A_{\mathcal G(k)},\lambda_{ji}(k),j,i\in\mathcal V \rangle, k\ge0\}$ is an independent process, then (\ref{ssaoidjaoi90e0}) holds if  there exist an  integer~$h>0$ such that
\bnaa\label{asdasdqwra}
\inf_{m\ge0}\Lambda_m^h>C_2 (C_3)^h\max\{1,C_1\}\sup_{m\ge0}\Bigg[\sum_{k=mh}^{(m+1)h-1}\sum_{q=0}^d\|\mathbb E[\overline A(k,q)]\|\frac{[(1+\psi_2)^q-1]}{2-(1+\psi_2)^q}\Bigg],
\enaa
 and $b(0)\le f_{C_1,\beta_a,\beta_H,N,d}(\psi_2)$, with $\psi_2\in(0,{2}^{\frac{1}{d}}-1)$,
where $C_1$ is defined in Condition \textbf{C1.d}.
}
\end{corollary}


\vskip 0.2cm

\begin{remark}
\rm{Theorem \ref{e4t5634332}, Corollaries \ref{tuiasdwqrqacasdwl}-\ref{tuilunasdasaw}  give  explicit convergence conditions under which all nodes' estimates converge to the true parameter in mean square.  Existing literature used the Lyapunov-Krasovskii functional method to deal with time delays and obtained the non-explicit LMI type convergence condition (\cite{Mill}).  In contrast, here,
we transform the system with random time-varying communication delays into an equivalent delay-free system by introducing an auxiliary system and then adopt the method of binomial expansion of  random matrix products to transform the mean square convergence analysis of the delay-free system into that of the mathematical expectation of random  matrix products, and obtain the key convergence conditions (\ref{saijdq98e00sad})-(\ref{ssaoidjaoi90e0}) which explicitly rely on the conditional expectations of delay matrices, observation matrices and weighted adjacency matrices of communication graphs   over a sequence of fixed-length time intervals.  In the absence of time delays,   the condition (\ref{saijdq98e00sad}) degenerates  to the condition (b.1) in Theorem \ref{nodelay111}.}
\end{remark}

\vskip 0.2cm

\begin{remark}\label{asd8q9280120}
\rm{The conditions (\ref{926hao123}) and (\ref{asdasdqwra}) can be further simplified  for special delay processes.
If the delays are independent of the graphs, then $\mathbb E[\overline A(k,q)]=\mathbb E[\mathcal A_{\mathcal G(k)}]\circ \mathbb E[\mathcal I(k,q)]$. Here, the element in the $i$th row and the $j$th column of  $\mathbb E[\mathcal I(k,q)]$, $\mathbb E[\mathcal I(k,q)]_{ij}=\mathbb P\{\lambda_{ji}(k)=q\}=p_{ji,q}(k)$. In addition,
 \begin{itemize}
   \item if  $\lambda_{ji}(k)$ are identically distributed w.r.t.  $k$, then  $\mathbb E[\mathcal I(k,q)]_{ij}=p_{ji,q}(0),\forall~k\ge0$;
   \item if  $\lambda_{ji}(k)$ are identically distributed w.r.t. both $k$ and $(j,i)$, then $\mathbb E[\mathcal I(k,q)]_{ij}= p_{q},i\not=j$ where $p_{q}$ denotes the probability that the packet is delayed by $q$ steps for all $k$ and $(j,i),j\not=i$. Therefore, $\|\mathbb E[\overline A(k,q)]\|=p_q\|\mathbb E[\mathcal A_{\mathcal G(k)}]\|$. Furthermore, if the graph sequence is an i.i.d. process, then the condition (\ref{asdasdqwra}) becomes
       \bann\label{asdsaswra}
\inf_{m\ge0}\Lambda_m^h>C_2 (C_3)^h\max\{1,C_1\}h\|\mathbb E[\mathcal A_{\mathcal G(0)}]\|\sum_{q=0}^d\frac{p_q[(1+\psi_2)^q-1]}{2-(1+\psi_2)^q}.
\eann
\end{itemize}
}
\end{remark}

\vskip 0.2cm

Corollaries \ref{tuiasdwqrqacasdwl}-\ref{tuilunasdasaw} show that for given algorithm gains $\{a(k), k\geq0\}$ and $\{b(k), k\geq0\}$, if
the communication graphs and observation matrices are persistently excited with
enough intensity, then  the additional effects of time delays can be mitigated. The maximum delay bound $d$ that can be allowed is related to  the weighted adjacency matrix of mean graphs $\mathbb E[\mathcal A_{\mathcal G(k)}]$, the probability distribution of time delays $\mathbb E[\mathcal I(k,q)]$ and the algorithm gains.  In the absence of time delays,
(\ref{ssaoidjaoi90e0})  degenerates to the condition (c.1) in Theorem \ref{nodelay}.   The following corollary shows that for the case with  conditionally balanced graphs, if the \emph{stochastic spatio-temporal persistence of excitation} condition $\inf_{m\ge0}{{\Lambda_m^h}}\ge \theta$ a.s.  holds, then for any given bounded delays, mean square convergence of the algorithm can be guaranteed if the algorithms gains are properly designed and sufficiently small.

\begin{corollary}\label{tuilunl}
\rm{ Suppose that Assumptions~\textbf{A1.a}--\textbf{A1.b}, \textbf{A2.b} hold, $\{\mathcal G(k), k\ge0\}\in\Gamma_1$ and   there exists an   integer $h>0$, a  constant $\theta>0$ such that  $\inf_{m\ge0}{{\Lambda_m^h}}\ge \theta$ a.s. If  Conditions \textbf{C1.a}--\textbf{C1.b} hold, $b(k)=O(a(k))$,   and $b(0)\le f_{C_1,\beta_a,\beta_H,N,d}(\psi_3)$ with $\psi_3\in (0,(1+{\theta}/[{\theta+{N}C_2 (C_3)^h\max\{1,C_1\}\beta_adh}])^{\frac{1}{d}}-1)$,
 then the  algorithm~(\ref{asapp}) converges in mean square.}
\end{corollary}

\section{numerical example}\label{zhen}
We apply our algorithm to  decentralized multi-area online state estimation in power systems to illustrate the effectiveness of the obtained theoretical results. An IEEE 14-bus system is used for the test, which has 14 buses and is partitioned into 4 areas~$A_1,A_2,A_3,A_4$, shown in Figure~\ref{wjd12ml}.
After a DC power flow approximation (\cite{lll}), the  grid state to be estimated degenerates into a vector of voltage  phase angles at all buses.  Let bus 1's    voltage phase angle   be zero, as the reference bus.   The  grid state to be estimated is given by
\bann
&&x_0=
[-4.98,-12.72,-11.33,-8.78,-14.22,-13.37,\cr
&&~~~~~~~~~~~~~~~~~~~~~~~-13.36,-14.94,-15.10,-14.79,-15.05,-15.12,-16.03]^T.
\eann
The measurements   $z_i(k)$ are linearly related to $x_0$, given by
$
z_i(k)=s_i(k)H_i'x_0+v_i(k),\ i=1, 2, 3, 4,
$
where the noise $\{v_i(k),k\ge0\}$ is assumed to be an i.i.d. process with the  standard normal distribution, $\{s_i(k),k\ge0\}$ is an i.i.d. sequence, modelling the sensing failures  with  $\mathbb P\{s_i(k)=1\}=\mathbb P\{s_i(k)=0\}=0.5$, and $H_i'$, $i=1,...,4$ are the observation matrices, which are deterministic and given in Appendix D. There are  $4$ random communication links with $0-1$ weights, represented by the red dotted lines in Figure~\ref{wjd12ml}. At odd time instants, the link from $A_2$ to $A_3$ awakes with the probability $0.5$ and the others sleep; at even time instants,  the link from $A_2$ to $ A_3$ sleeps and the others  awake  with the probability $0.5$. Both $\{\mathcal G(k),k\ge0\}$ and $\{H(k),k\ge0\}$  are independent processes. We use the averaged relative error, $\frac{\sum_{i=1}^4\|x_i(k)-x_0\|}{4\|x_0\|}$, to evaluate the performance of the algorithm.

\vspace{-0.5cm}
\begin{figure}[H]
  \centering
  \includegraphics[width=6cm]{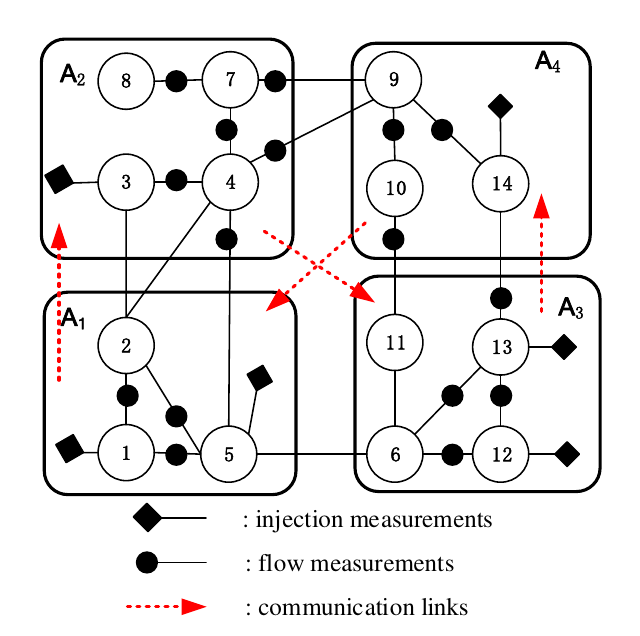}\\
\setlength{\abovecaptionskip}{0pt}
\setlength{\belowcaptionskip}{0pt}
  \caption{{IEEE-14 multi-area buses and the communication graphs.}}\label{wjd12ml}
\end{figure}
 \vspace{-0.7cm}

   For the delay-free case, set $a(k)=b(k)=\frac{0.5}{(k+1)^{0.52}}$.  Let $c(m)=\frac{0.0112}{(2m+2)^{0.52}}$. When $h=2$, we plot the curves of $\overline\Lambda_m^2$ and $c(m)$ w.r.t. $m$ in Figure \ref{delayfreecondition}, which shows that $\overline\Lambda_m^2\ge c(m),m\ge0$. The conditions of Theorem \ref{nodelay111} hold.  Figure~\ref{fig1} is depicted with the curves of the  averaged relative errors, where the red line represents the error curve of the algorithm  without random link failures and sensing failures, as the base case. It shows that in spite of  the unbalance of the mean graphs and the sensing failures, the four areas' estimates  converge to~$x_0$.

For the case with time delays,  assume that the delays are  independent of the communication graphs, observation matrices and measurement noises, and subjected to the Bernoulli distribution, i.e.
$\lambda_{ji}(k)\sim B(d,p)$ for all $k$ and $(j,i)$.
 Then~\bnaa\label{sadqwq9qsllll}
 \mathbb P\{\lambda_{ji}(k)=q\}=\mathbb C_{d}^qp^q(1-p)^{d-q},q=0,\cdots,d.
 \enaa
Set $d=4,p=0.4$. We now verify the convergence conditions in Corollary \ref{tuiasdwqrqacasdwl}. Let $a(k)=b(k)$. Then $C_1=1$. By the above settings of communication graphs and measurement matrices, we know that $\beta_a=1$,$\beta_H=4.07$. Let $\psi_2=0.01$. Then we have  $f_{C_1,\beta_a,\beta_H,N,d}(\psi_2)= 0.0005$. Then, let $b(k)=\frac{ 0.0005}{(k+1)^{0.1}}$.
By the definition of  $p_q$ in Remark \ref{asd8q9280120},  it follows from (\ref{sadqwq9qsllll}) that $p_q=\mathbb C_{4}^q0.4^q0.6^{4-q},q=0,...,4$.
Note that $\|\mathbb E[\mathcal A_{\mathcal G(k)}]\|\equiv0.5$. As is discussed in Remark \ref{asd8q9280120}, $\|\mathbb E[\overline A(k,q)]\|=p_q\|\mathbb E[\mathcal A_{\mathcal G(k)}]\|$. Hence, it can be calculated that
\bann
\sum_{k=mh}^{(m+1)h-1}b(k)\sum_{q=0}^d\|\mathbb E[\overline A(k,q)]\|\frac{[(1+\psi_2)^q-1]}{2-(1+\psi_2)^q}&=&\sum_{k=2m}^{2m+1}b(k)\sum_{q=0}^4\frac{0.5p_q[(1+\psi_2)^q-1]}{2-(1+\psi_2)^q}
\cr
&=&0.01\sum_{k=2m}^{2m+1}b(k).
\eann
Note that
$\overline{\Lambda}_m^2=\lambda_{\min}[\sum_{k=2m}^{2m+1}b(k)(\mathbb E[\widehat {\mathcal L}_{\mathcal G(k)}]\otimes I_{13}+\mathbb E[{\mathcal H}^T(k){\mathcal H}(k)])]$. Let $c(m)=\frac{ 0.00001}{(4m+4)^{0.1}}$. We plot the curves of $\overline{\Lambda}_m^2-0.01\sum_{k=2m}^{2m+1}b(k)$ and $c(m)$ w.r.t. $m$ in Figure \ref{saspjaksdaka}, showing that the condition (\ref{926hao123}) in Corollary \ref{tuiasdwqrqacasdwl} holds.
Figure \ref{theoremV1} is depicted with curve of the averaged relative error, which confirms  Corollary \ref{tuiasdwqrqacasdwl}.

\begin{figure}[H]
\begin{minipage}[t]{0.5\linewidth}
  \centering
  \includegraphics[width=8cm]{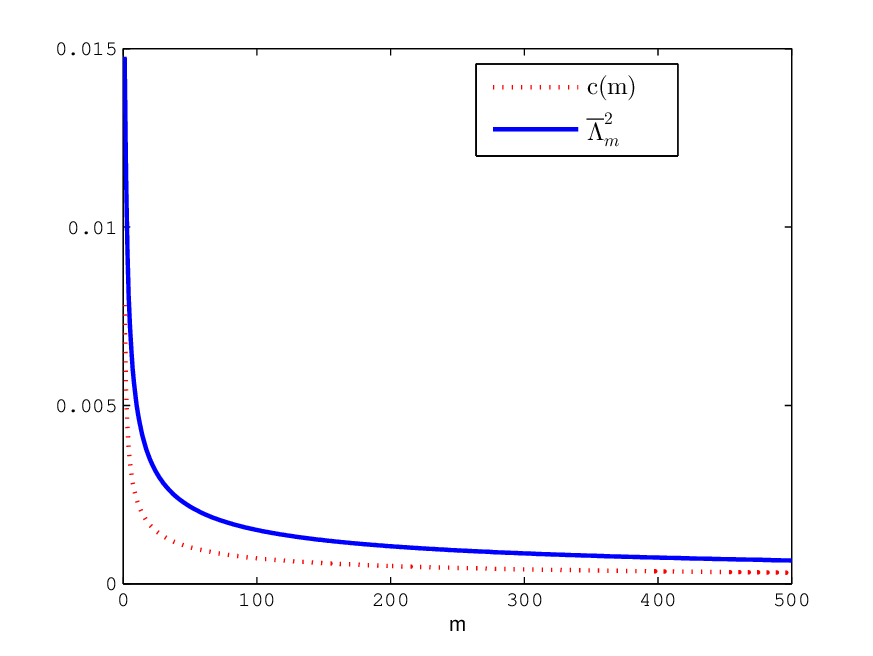}\\
\setlength{\abovecaptionskip}{0pt}
\setlength{\belowcaptionskip}{0pt}
  \caption{\footnotesize{Curves of  ${\overline{\Lambda}_m^2}$ and $c(m)=\frac{0.0112}{(2m+2)^{0.52}}$ w.r.t. $m$. }}\label{delayfreecondition}
\end{minipage}
\begin{minipage}[t]{0.5\linewidth}
  \centering
  \includegraphics[width=8cm]{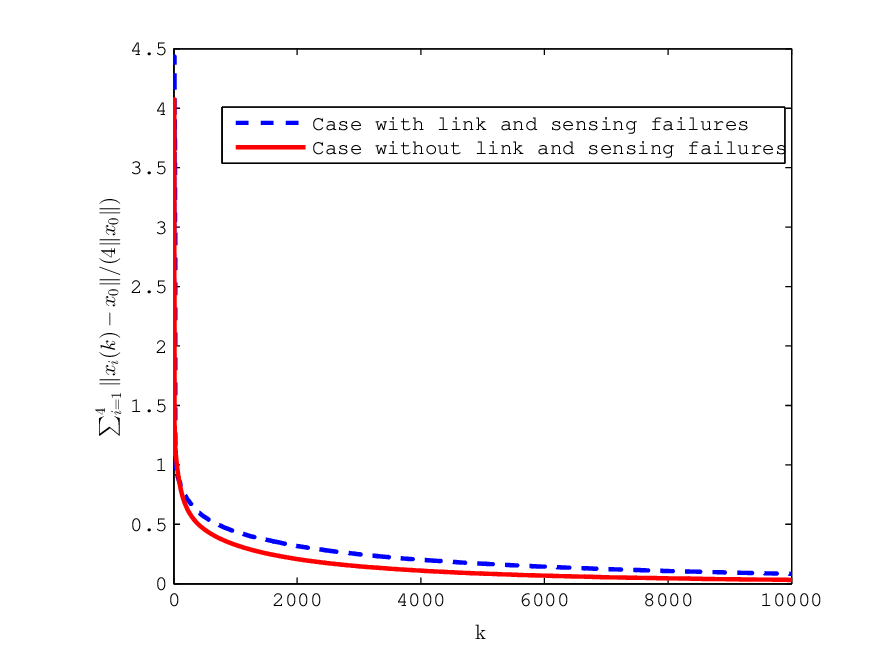}\\
\setlength{\abovecaptionskip}{0pt}
\setlength{\belowcaptionskip}{0pt}
  \caption{\footnotesize{Curves of the averaged relative  error $\frac{\sum_{i=1}^4\|x_i(k)-x_0\|}{4\|x_0\|}$ for the delay-free case. }}\label{fig1}
\end{minipage}
\end{figure}

\vspace{-1cm}

\begin{figure}[H]
\begin{minipage}[t]{0.5\linewidth}
  \centering
  \includegraphics[width=8cm]{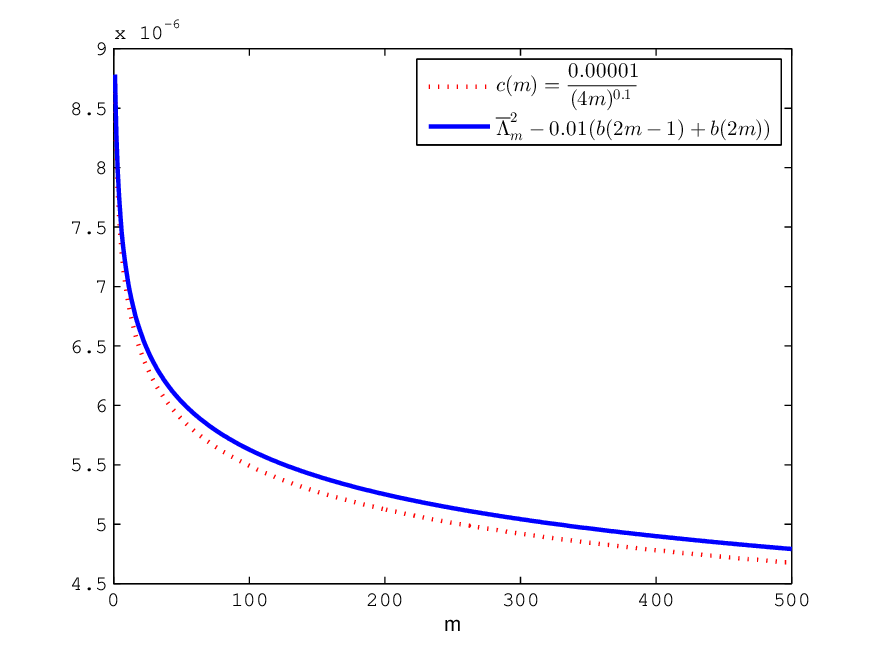}\\
\setlength{\abovecaptionskip}{0pt}
\setlength{\belowcaptionskip}{0pt}
  \caption{\footnotesize{{Curves of $\overline{\Lambda}_m^2-0.01\sum_{k=2m}^{2m+1}b(k)$ and $c(m)$ w.r.t.  $m$. }}}\label{saspjaksdaka}
  \end{minipage}
\begin{minipage}[t]{0.5\linewidth}
  \centering
  \includegraphics[width=8cm]{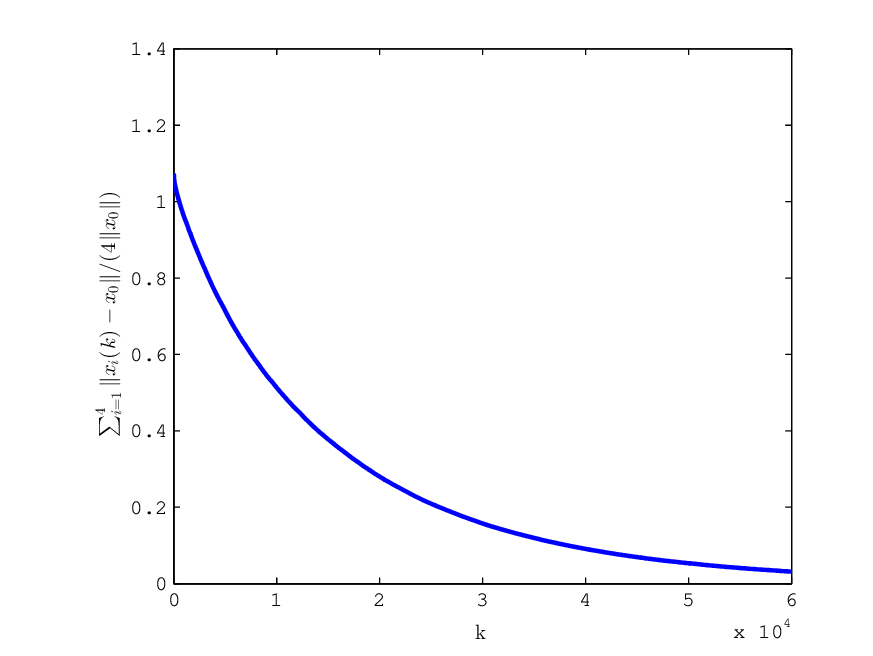}\\
\setlength{\abovecaptionskip}{0pt}
\setlength{\belowcaptionskip}{0pt}
  \caption{\footnotesize{{Curve of the averaged relative  error  for the case with random time delays. }}}\label{theoremV1}
  \end{minipage}
\end{figure}

\section{conclusion}\label{conc}
In this paper, we analyzed the convergence of the decentralized  cooperative online parameter estimation algorithm in an uncertain communication environment. Each node has a partial linear observation of the unknown parameter with random time-varying observation matrices.
The underlying communication network is modeled by a sequence of random digraphs and is subjected to nonuniform random time-varying delays in channels.
  For the delay-free case, we proved that if the observation matrices and the graph sequence satisfy the \emph{stochastic spatio-temporal persistence of excitation} condition, then the algorithm gains can be designed properly such that all nodes' estimates  converge to the true parameter in mean square and almost surely. Specially, for  Markovian switching  communication graphs and observation matrices, this condition holds if the stationary graph is balanced with a spanning tree and the measurement model is spatio-temporally jointly observable. For the case with communication delays, we introduced  delay matrices to model the random time-varying communication delays, adopted the method of binomial expansion of  random matrix products to transform the mean square convergence analysis of the algorithm into that of the mathematical expectation of   random  matrix products, and obtained  mean square convergence conditions explicitly relying on the conditional expectations of delay matrices, observation matrices and weighted adjacency matrices of communication graphs  over a sequence of fixed-length intervals. In the absence of time delays, these mean square convergence conditions degenerate to the  \emph{stochastic spatio-temporal persistence of excitation} conditions. Especially, given that the digraphs are conditionally balanced, we  show that if the \emph{stochastic spatio-temporal persistence of excitation} condition holds,  then for any given   bounded delay, proper algorithm gains can be designed to guarantee mean square convergence of the algorithm.

There are many interesting open issues for future research. Theorem  \ref{e4t5634332} is established for a very general type of delays, namely random and unordered. This means that in the practical implementation, the packets of information exchanged by pairs of nodes are being placed in a processing queue without any regard to their transmit time stamp. In some cases, all received packets are ordered by the time stamp of their transmission, and the communication delays would be random and monotone (\cite{ouis},\cite{Fsadawasa},\cite{Fsadawasa1}).  How to explore monotonicity constraints in the random delay process to relax the conditions or strengthen the results of Theorem V.1 would be an interesting and challenging issue.  The main obstacle is how to deal with the delay-induced products of the inverse of matrices, which is difficult and may need more advanced techniques. Another important issue is the convergence rate of the algorithm. Especially, Corollary \ref{tuilunl}  shows that for the case with  conditionally balanced graphs, if the \emph{stochastic spatio-temporal persistence of excitation} condition  holds, then for any given bounded delays, mean square convergence of the algorithm can be guaranteed if we choose  sufficiently small algorithms gains. However, smaller algorithm gains generally lead to a slower convergence. Thus, how to choose the algorithms gains for optimizing the convergence rate is an interesting topic for future investigation.

\begin{appendices}
\section{several useful lemmas}
\begin{definition}(\cite{053})\label{sadkasdd90we}
\rm{A Markov chain on a countable state space $\mathcal S$ with a stationary distribution $\pi$ and transition function $\mathbb P(x,\cdot)$  is called uniform ergodic, if there exist positive constants $r>1$ and $R$ such that for all $x\in\mathcal S$,
$\|\mathbb P^n(x,\cdot)-\pi\|\le Rr^{-n}$.
Here, $\|\mathbb P^n(x,\cdot)-\pi\|=\sum_{y}|\mathbb P^n(x,y)-\pi_y|$.}
\end{definition}

\begin{lemma}\label{98pkk0}
(\cite{9jko00})
\rm{For any given matrix~$P$, denote~$W=I-P$.  If there exists a constant~$\psi\in(0,1)$ such that~$\|P\|\le \psi$, then~$W$ is invertible and $\|W^{-1}\|\le(1-\|P\|)^{-1}\le(1-\psi)^{-1}.$}
\end{lemma}

\begin{lemma}\label{guotime}
(\cite{guo})
\rm{Assume that~$\{s_1(k),k\ge0\}$ and $\{s_2(k),k\ge0\}$ are real  sequences satisfying $0\le s_2(k)<1$, $\sum_{k=0}^\infty s_2(k)=\infty$ and $\lim_{k\to \infty}\frac{s_1(k)}{s_2(k)}$ exists. Then
\bann
\lim_{k\to\infty}\sum_{i=1}^{k}s_1(i)\prod_{l=i+1}^k(1-s_2(l))=\lim_{k\to\infty}\frac{s_1(k)}{s_2(k)}.
\eann}
\end{lemma}

\begin{lemma}(\cite{rb})\label{robbins}
\rm{Assume~that $\{x(k),\mathcal F(k)\},~\{\a(k),\mathcal F(k)\},~\{\b(k),\mathcal F(k)\}~\mathrm{and}~\{\gamma(k),\mathcal F(k)\}$ are all nonnegative adaptive sequences, satisfying
$$
\mathbb E[x(k+1)|\mathcal F(k)]\le(1+\a(k))x(k)-\b(k)+\gamma(k),k\ge0~\rm{a.s.}
$$
If~$\sum_{k=0}^\infty(\a(k)+\gamma(k))<\infty~\rm{a.s.}$, then~$x(k)$ converges to a finite random variable~\rm{a.s.} and~$\sum_{k=0}^\infty \b(k)<\infty~\rm{a.s.}$}
\end{lemma}

For the subsequent Lemmas \ref{lyapunov} and \ref{holder}, the readers may be referred to Theorem
6.4 and its next paragraph in Ch. 6 of \cite{Foundatio}.
\begin{lemma}(Conditional Lyapunov inequality)\label{lyapunov}
\rm{Denote the probability space by~$(\Psi,\mathcal F,P)$. Let~$\mathcal F_1$ be a sub $\sigma-$algebra of $\mathcal F$
and~$\xi$ be a random variable on~$(\Psi,\mathcal F,P)$.  Then
$(\mathbb E[|\xi|^s|\mathcal F_1])^{\frac{1}{s}}\le(\mathbb E[|\xi|^t|\mathcal F_1])^{\frac{1}{t}}~\mathrm{a.s.},~0<s<t.$}
\end{lemma}

\begin{lemma}(Conditional~H\"{o}lder inequality)\label{holder}
\rm{Denote the probability space~$(\Psi,\mathcal F,P)$.  Let~$\mathcal F_1$ be a sub $\sigma-$algebra of $\mathcal F$. Let~$\xi$ and~$\eta$ be two random variables on~$(\Psi,\mathcal F,P)$. Let constants~$p\in(1,\infty),~q\in(1,\infty)$ and~$1/p+1/q=1$. If~$\mathbb E[|\xi|^p]<\infty$ and$~\mathbb E[|\eta|^q]<\infty$, then
$\mathbb E[|\xi\eta||\mathcal F_1]\le(\mathbb E[|\xi|^p|\mathcal F_1])^{\frac{1}{p}}(\mathbb E[|\eta|^q|\mathcal F_1])^{\frac{1}{q}}~\rm{a.s.}$}
\end{lemma}

\begin{lemma}\label{fanshuguji}
\rm{For any random matrix $A\in\mathbb R^{m\times n}$, $\|\mathbb E[AA^T]\|\le n\|\mathbb E[A^TA]\|$.}
\end{lemma}
\begin{proof}
By the properties of matrix trace, we have
$\|\mathbb E[AA^T]\|=\lambda_{\max}(\mathbb E[AA^T])\le \mathrm{Tr}(\mathbb E[AA^T])=\mathrm{Tr}(\mathbb E[A^TA])\le n\lambda_{\max}(\mathbb E[A^TA])
=n\|\mathbb E[A^TA]\|$.
\end{proof}

\begin{lemma}
\label{keyconnectivitylemmaprepareadd}
\rm{Let $\mathcal{A}=[a_{ij}]_{N\times N}$ be a weighted adjacency matrix of an undirected graph with $N$ nodes and $\mathcal{L}$ be the associated Laplacian matrix. Let $x=[x^T_{1},...,x_{N}^T]^T\in\mathbb{R}^{Nn}$ be any given nonzero $Nn$-dimensional vector where $x_{i}\in\mathbb{R}^{n}$, $i=1,2,...,N$ and there exists $i\neq j$, such that $x_i\neq x_j$. If $a_{ij}\geq0$, $i,j=1,2,...,N$ and the graph is connected, then $x^T(\mathcal{L}\otimes I_{n})x>0$.}
\end{lemma}
\begin{proof}  By the definition of Laplacian matrix, we have $x^T(\mathcal{L}\otimes I_{n})x=\frac{1}{2}\sum_{i=1}^{N}\sum_{j=1}^{N}a_{ij}\|x_{i}-x_{j}\|^2$. Noting that there exists $i\neq j$, such that $x_i\neq x_j$ and the graph is connected, by $a_{ij}\geq0$, $i,j=1,2,...,N$,  we get $x^T(\mathcal{L}\otimes I_{n})x>0$.
\end{proof}

\section{proofs in Section \ref{jieguo}}
Let
\bnaa\label{gfkasdasad}
P(k)=I_{Nn}-D(k),
\enaa
where
\bnaa\label{gkexpre}
D(k)=b(k)\mathcal L_{\mathcal G(k)}\otimes I_n+a(k){\mathcal H}^T(k){\mathcal H}(k).
\enaa
The proof of Theorem~\ref{nodelay111} needs the following lemma.
\begin{lemma}\label{mnbhgyls12}
\rm{For the algorithm~(\ref{asapp}), if Condition \textbf{C1.a}, the conditions (b.1) and   (b.2) in Theorem \ref{nodelay111} hold, 
then
\bnaa\label{ksaopS1}
\lim_{k\to\infty}\|\mathbb E[\Phi_P(k,0)\Phi_P^T(k,0)]\|=0.
\enaa}
\end{lemma}
\begin{proof}
By (\ref{gfkasdasad}),  we have
\bnaa\label{lesisyifdiosjd}
&&\Phi^T_P((m+1)h-1,mh)\Phi_P((m+1)h-1,mh)\cr
&=&(I_{Nn}-D^T(mh))\cdots(I_{Nn}-D^T((m+1)h-1))\cr
&&\times(I_{Nn}-D((m+1)h-1))\cdots(I_{Nn}-D(mh)).
\enaa
Taking conditional expectation w.r.t.~$\mathcal F(mh-1)$ on both sides of the above, by the binomial expansion, we have
\bna\label{1712323}
&&~~~\|\mathbb E[\Phi^T_P((m+1)h-1,mh)\Phi_P((m+1)h-1,mh)|\mathcal F(mh-1)]\|\cr
&&=\|\mathbb E[(I_{Nn}-D^T(mh))\cdots(I_{Nn}-D^T((m+1)h-1))\cr
&&~\times(I_{Nn}-D((m+1)h-1))\cdots(I_{Nn}-D(mh))|\mathcal F(mh-1)]\|\cr
&&=\Big\|I_{Nn}-\sum_{k=mh}^{(m+1)h-1}\mathbb E[D^T(k)+D(k)|\mathcal F(mh-1)]\cr
&&~+\mathbb E[M_2(m)+\cdots+M_{2h}(m)|\mathcal F(mh-1)]\Big\|\cr
&&\le\Big\|I_{Nn}-\sum_{k=mh}^{(m+1)h-1}\mathbb E[D^T(k)+D(k)|\mathcal F(mh-1)]\Big\|\cr
&&~+\left\|\mathbb E[M_2(m)+\cdots+M_{2h}(m)|\mathcal F(mh-1)]\right\|.
\ena
Here,~$M_i(m),i=2,\cdots,2h$ represent the $i$-th order terms in the binomial expansion of $\Phi_P((m+1)h-1,mh)\Phi_P^T((m+1)h-1,mh)$.

Since the 2-norm of a symmetric matrix  is equal  to its spectral radius, by the definition of spectral radius, we have
\bna\label{a0xdrtg}
&&~~~\Big\|I_{Nn}-\sum_{k=mh}^{(m+1)h-1}\mathbb E[D(k)+D^T(k)|\mathcal F(mh-1)]\Big\|\cr
&&=\rho\Bigg(I_{Nn}-\sum_{k=mh}^{(m+1)h-1}\mathbb E[D(k)+D^T(k)|\mathcal F(mh-1)]\Bigg)\cr
&&=\max_{1\le i\le Nn}\Bigg|\lambda_i\Bigg(I_{Nn}-\sum_{k=mh}^{(m+1)h-1}\mathbb E[D(k)+D^T(k)|\mathcal F(mh-1)]\Bigg)\Bigg|\cr
&&=\max_{1\le i\le Nn}\Bigg|1-\lambda_i\Bigg(\sum_{k=mh}^{(m+1)h-1}\mathbb E[D(k)+D^T(k)|\mathcal F(mh-1)]\Bigg)\Bigg|.
\ena
Since both $a(k)$ and $b(k)$ tend to zero, by  the condition  (b.2),  we   know that there exists a positive integer~$m_1$, which is independent of the sample paths, such that
\bann\label{lambda1}
\lambda_i\Bigg(\sum_{k=mh}^{(m+1)h-1}\mathbb E[D(k)+D^T(k)|\mathcal F(mh-1)]\Bigg)\le1,~i=1,\cdots,Nn,~\forall~m\ge m_1~\rm{a.s.}
\eann
This together with   (\ref{1712323}) and~(\ref{a0xdrtg}) leads to
\bnaa\label{a098okmsadsasd}
&&~~~\|\mathbb E[\Phi_P^T((m+1)h-1,mh)\Phi_P((m+1)h-1,mh)|\mathcal F(mh-1)]\|\cr
&&\le1-\lambda_{\min}\Bigg(\sum_{k=mh}^{(m+1)h-1}\mathbb E[D(k)+D^T(k)|\mathcal F(mh-1)]\Bigg)\cr
&&~+\left\|\mathbb E[M_2(m)+\cdots+M_{2h}(m)|\mathcal F(mh-1)]\right\|,~\forall~m\ge m_1~\rm{a.s.}
\enaa
We next bound the two terms on the right side of the above. For the first term, by  the definitions of~$D(k)$ and~$\overline{\Lambda}_m^h$ and the condition (b.1),  we have
\bna\label{a098okssam1s}
&&~~~1-\lambda_{\min}\Bigg(\sum_{k=mh}^{(m+1)h-1}\mathbb E[D(k)+D^T(k)|\mathcal F(mh-1)]\Bigg)\cr
~&&=1-\lambda_{\min}\Bigg(\sum_{k=mh}^{(m+1)h-1}\mathbb E[2b(k)\widehat{\mathcal L}_{\mathcal G(k)}\otimes I_n+2a(k){\mathcal H}^T(k){\mathcal H}(k)|\mathcal F(mh-1)]\Bigg)\cr
~&&=1-2\overline{\Lambda}_m^h\cr
~&&\le1-c(m),~\forall~m\ge m_1~\rm{a.s.}
\enaa
By  Lemma \ref{lyapunov}  and the condition (b.2), it follows that
\bann\label{sadjqwweda}
&&\sup_{k\ge0}\mathbb E[\|\widetilde D(k)\|^i|\mathcal F(k-1)]\le\sup_{k\ge0}[\mathbb E[\|\widetilde D(k)\|^{2^h}|\mathcal F(k-1)]]^{\frac{i}{2^h}}\le\rho_0^i~\mathrm{a.s.},~2\le i\le 2^h,
\eann
where~$\widetilde D(k)=\mathcal L_{\mathcal G(k)}\otimes I_n+{\mathcal H}^T(k){\mathcal H}(k)$. Note that for any given random variable~$\xi$ and~$\sigma$-algebra~$\mathcal F_1\subseteq\mathcal F_2,$ it is true that
 \bnaa\label{klkjl}
 \mathbb E[\xi|\mathcal F_1]=\mathbb E[\mathbb E[\xi|\mathcal F_2]|\mathcal F_1].
 \enaa
We then have $$\mathbb E[\|\widetilde D(k)\|^l|\mathcal F(mh-1)]=\mathbb E[\mathbb E[\|\widetilde D(k)\|^l|\mathcal F(k-1)]|\mathcal F(mh-1)],~2\le l\le 2^h,~k\ge mh.$$
From  the  definitions of~$M_i(m),i=2,\cdots,2h$ and the above, by termwise multiplication and using
Lemma \ref{holder}   repeatedly, for the second term on the right side of (\ref{a098okmsadsasd}),  we have
\bnaa\label{adlfjoqw}
\|\mathbb E[M_2(m)+\cdots+M_{2h}(m)|\mathcal F(mh-1)]\|&\le& b^2(mh)\Bigg(\sum_{i=2}^{2h}\mathbb C_{2h}^i(\max\{1,\phi\}\rho_0)^i\Bigg)\cr
&=&b^2(mh)\a,
\enaa
where $\phi$ satisfies $a(k)\le \phi b(k)$,  $\a=(1+\max\{1,\phi\}\rho_0)^{2h}-1-2h\max\{1,\phi\}\rho_0$ and $\mathbb C_m^p$ denotes the combinatorial number of choosing $p$ elements from $m$ elements.
By~(\ref{a098okmsadsasd})-(\ref{adlfjoqw}), we have
\bna\label{0oooopw111}
&&~~~\|\mathbb E[\Phi_P^T((m+1)h-1,mh)\Phi_P((m+1)h-1,mh)|\mathcal F(mh-1)]\|\cr
~&&\le1-c(m)+b^2(mh)\a,~m\ge m_1~\rm{a.s.}
\ena

Denote $m_k=\lfloor \frac{k}{h} \rfloor$. By the properties of the conditional expectation,  Lemma \ref{fanshuguji} and (\ref{0oooopw111}), we have
\bnaa\label{xasoad}
&&~~~\|\mathbb E[\Phi_P(k,0)\Phi_P^T(k,0)]\|\cr
&&\le Nn\|\mathbb E[\Phi_P^T(k,0)\Phi_P(k,0)]\|\cr
&&=Nn\|\mathbb E[\Phi_P^T(m_kh-1,0)\Phi_P^T(k,m_kh)\Phi_P(k,m_kh)\Phi_P(m_kh-1,0) ]\|\cr
&&\le Nn\|\mathbb E[\Phi_P^T(m_kh-1,0)\|\Phi_P^T(k,m_kh)\Phi_P(k,m_kh)\|\Phi_P(m_kh-1,0) ]\|\cr
&&=Nn\|\mathbb E[\mathbb E[\Phi_P^T(m_kh-1,0)\|\Phi_P^T(k,m_kh)\Phi_P(k,m_kh)\|\Phi_P(m_kh-1,0) |\mathcal F(m_kh-1)]]\|\cr
&&=Nn\|\mathbb E[\Phi_P^T(m_kh-1,0)\mathbb E[\|\Phi_P^T(k,m_kh)\Phi_P(k,m_kh)\||\mathcal F(m_kh-1)]\cr
&&~\times\Phi_P(m_kh-1,0) ]\|,
\enaa
For any positive  integers   $m,n$ satisfying $0\leq m-n\le h-1,$ it follows from   the condition (b.2) that there exists a constant $\rho^{*}_{h}>0$ such that
\bnaa\label{xiaoyuh}
\|\mathbb E[\Phi_P^T(m,n)\Phi_P(m,n)|\mathcal F(n-1)]\|<\rho^{*}_{h} \mathrm{a.s.}
\enaa
By the above and (\ref{xasoad}), noting that $k-m_kh\le h-1$, we have
\bnaa\label{wjds}
&&~~~\|\mathbb E[\Phi_P(k,0)\Phi_P^T(k,0)]\|~~~~~~~~~~~~~~~~~~~~~~\cr
&&\le\rho^{*}_{h}Nn\|\mathbb E[\Phi_P^T(m_kh-1,0)\Phi_P(m_kh-1,0) ]\|\cr
&&=\rho^{*}_{h}Nn\|\mathbb E[\Phi_P^T(m_1h-1,0)\Phi_P^T(m_kh-1,m_1h)\Phi_P(m_kh-1,m_1h)\Phi_P(m_1h-1,0)]\|\cr
&&=\rho^{*}_{h}Nn\|\mathbb E[\mathbb E(\Phi_P^T(m_1h-1,0)\Phi_P^T(m_kh-1,m_1h)\cr
&&~\times\Phi_P(m_kh-1,m_1h)\Phi_P(m_1h-1,0)|\mathcal F(m_1h-1))]\|\cr
&&\le\rho^{*}_{h}Nn\|\mathbb E[\Phi_P^T(m_1h-1,0)\cr
&&~\times\|\mathbb E[\Phi_P^T(m_kh-1,m_1h)\Phi_P(m_kh-1,m_1h)|\mathcal F(m_1h-1)]\|\Phi_P(m_1h-1,0)]\|.
\enaa
By (\ref{klkjl}) and (\ref{0oooopw111}),  we have
\bnaa\label{lianchengji}
&&~~~\|\mathbb E[\Phi_P^T(m_kh-1,m_1h)\Phi_P(m_kh-1,m_1h)|\mathcal F(m_1h-1)]\|~~~~~~~~~~~~~~~~\cr
&&=\|\mathbb E[\Phi_P^T((m_k-1)h-1,m_1h)\Phi_P^T(m_kh-1,(m_k-1)h)\Phi_P(m_kh-1,(m_k-1)h)\cr
&&~~~~~~\times\Phi_P((m_k-1)h-1,m_1h)|\mathcal F(m_1h-1)]\|\cr
&&=\|\mathbb E[\mathbb E[\Phi_P^T((m_k-1)h-1,m_1h)\Phi_P^T(m_kh-1,(m_k-1)h)\Phi_P(m_kh-1,(m_k-1)h)\cr
&&~~~~~~\times\Phi_P((m_k-1)h-1,m_1h)|\mathcal F((m_k-1)h-1)]|\mathcal F(m_1h-1)]\|\cr
&&=\|\mathbb E[\Phi_P^T((m_k-1)h-1,m_1h)\cr
&&~~~~~~\times\mathbb E[\Phi_P^T(m_kh-1,(m_k-1)h)\Phi_P(m_kh-1,(m_k-1)h)|\mathcal F((m_k-1)h-1)]\cr
&&~~~~\times\Phi_P((m_k-1)h-1,m_1h)|\mathcal F(m_1h-1)]\|\cr
&&\le\|\mathbb E[\Phi_P^T((m_k-1)h-1,m_1h)\cr
&&~~~~\times\|\mathbb E[\Phi_P^T(m_kh-1,(m_k-1)h)\Phi_P(m_kh-1,(m_k-1)h)|\mathcal F((m_k-1)h-1)]\|\cr
&&~~~~\times\Phi_P((m_k-1)h-1,m_1h)|\mathcal F(m_1h-1)]\|\cr
&&\le[1-c(m_k-1)+b^2((m_k-1)h)\a]\cr
&&~~~\times\|\mathbb E[\Phi_P^T((m_k-1)h-1,m_1h)\Phi_P((m_k-1)h-1,m_1h)|\mathcal F(m_1h-1)]\|\cr
&&\le\prod_{s=m_1}^{m_k-1}[1-c(s)+b^2(sh)\a]~\rm{a.s.},
\enaa
which together with~(\ref{wjds}) leads to
\bnaa\label{limit1}
&&~~~\|\mathbb E[\Phi_P(k,0)\Phi_P^T(k,0)]\|\cr
&&\le\rho^{*}_{h}Nn\|\mathbb E[\Phi_P^T(m_1h-1,0)\Phi_P(m_1h-1,0)]\|\prod_{s=m_1}^{m_k-1}[1-c(s)+b^2(sh)\a].
\enaa
By (\ref{cmsequence}), we know that there exists a positive integer $m_2$ such that
\bnaa\label{sizeab}
b^2(mh)\a\le\frac{1}{2}c(m),~\forall~m\ge m_2,
\enaa
Let $m_3=\max\{m_2,m_1\}$ and $r_1=\prod_{s=m_1}^{m_3-1}[1-c(s)+b^2(sh)\a]$. By (\ref{cmsequence}) and (\ref{sizeab}),
we have
\bnaa\label{lesisyu42}
&&~~~\lim_{k\to\infty}\prod_{s=m_1}^{m_k-1}[1-c(s)+b^2(sh)\a]\cr
&&\le\lim_{k\to\infty}r_1\prod_{s=m_3}^{m_k-1}[1-\frac{1}{2}c(s)]\cr
&&\le\lim_{k\to\infty}r_1\exp\Big(-\frac{1}{2}\sum_{s=m_3}^{m_k-1}c(s)\Big)\cr
&&=r_1\exp\Big(-\frac{1}{2}\sum_{s=m_3}^{\infty}c(s)\Big)=0.
\enaa

Since $\|\mathbb E[\Phi_P^T(m_1h-1,0)\Phi_P(m_1h-1,0)]\|<\infty$ by the condition (b.2),  (\ref{limit1}) and (\ref{lesisyu42}), we have (\ref{ksaopS1}). The lemma is proved.
\end{proof}

\vskip 0.2cm

\begin{proof}[\textbf{Proof of Theorem~\ref{nodelay111}}]
If $\lambda_{ji}(k)=0$ a.s., $\forall\ j,i\in\mathcal V$, $\forall\ k\ge0$, then by (\ref{zz9qjasoda}), we have
\bnaa\label{s1s121}
e(k+1)&=&P(k)e(k)+a(k){\mathcal H}^T(k)v(k)\cr
&=&\Phi_P(k,0)e(0)+\sum_{i=0}^ka(i)\Phi_P(k,i+1){\mathcal H}^T(i)v(i),~k\ge 0.
\enaa
By the above, we have
\bnaa\label{s1s11asfaxd}
\mathbb E[e(k+1)e^T(k+1)]&=&\mathbb E[\Phi_P(k,0)e(0)e^T(0)\Phi_P^T(k,0)]\cr
&&+\mathbb E\Bigg[\Phi_P(k,0)e(0)\sum_{i=0}^ka(i)[\Phi_P(k,i+1){\mathcal H}^T(i)v(i)]^T\Bigg]\cr
&&+\mathbb E\Bigg[\sum_{i=0}^ka(i)\Phi_P(k,i+1){\mathcal H}^T(i)v(i)[\Phi_P(k,0)e(0)]^T\Bigg]\cr
&&+\mathbb E\Bigg[\Bigg(\sum_{i=0}^ka(i)\Phi_P(k,i+1){\mathcal H}^T(i)v(i)\Bigg)\cr
&&\times\Bigg(\sum_{i=0}^ka(i)\Phi_P(k,i+1){\mathcal H}^T(i)v(i)\Bigg)^T\Bigg].
\enaa
By  Assumptions~\textbf{A1.a} and \textbf{A1.b},
we know that the second  and third terms on the right side of (\ref{s1s11asfaxd}) are both equal to zero. Moreover, from
\bnaa\label{jiaochaxiang}
\mathbb E[v(i)v^T(j)]=\mathbb E[\mathbb E[v(i)v^T(j)|\mathcal F(i-1)]]=\mathbb E[\mathbb E[v(i)|\mathcal F(i-1)]v^T(j)]=0,~\forall~i>j,
\enaa
we have
\bann
&&~~~\mathbb E\Bigg[\Bigg(\sum_{i=0}^ka(i)\Phi_P(k,i+1){\mathcal H}^T(i)v(i)\Bigg)\Bigg(\sum_{i=0}^ka(i)\Phi_P(k,i+1){\mathcal H}^T(i)v(i)\Bigg)^T\Bigg]\cr
&&=\mathbb E\Bigg[\sum_{i=0}^ka^2(i)\Phi_P(k,i+1){\mathcal H}^T(i)v(i)v^T(i){\mathcal H}(i)\Phi_P(k,i+1)\Bigg].
\eann
Substituting the above  into~(\ref{s1s11asfaxd}) and taking the 2-norm leads to
\bnaa\label{s1s}
&&~~~\|\mathbb E[e(k+1)e^T(k+1)]\|\cr
&&\le\|\mathbb E[\Phi_P(k,0)\Phi_P^T(k,0)]\|\|e(0)\|^2\cr
&&~~+\sum_{i=0}^ka^2(i)\|\mathbb E[\Phi_P(k,i+1){\mathcal H}^T(i)v(i)v^T(i){\mathcal H}(i)\Phi_P^T(k,i+1)]\|\cr
&&=\|\mathbb E[\Phi_P(k,0)\Phi_P^T(k,0)]\|\|e(0)\|^2\cr
&&~~+\sum_{i=k-3h}^{k}a^2(i)\|\mathbb E[\Phi_P(k,i+1){\mathcal H}^T(i)v(i)v^T(i){\mathcal H}(i)\Phi_P^T(k,i+1)]\|\cr
&&~~+\sum_{i=0}^{k-3h-1}a^2(i)\|\mathbb E[\Phi_P(k,i+1){\mathcal H}^T(i)v(i)v^T(i){\mathcal H}(i)\Phi_P^T(k,i+1)]\|.
\enaa
By Lemma~\ref{mnbhgyls12}, we know that the first term in the above converges to zero. For the second term in the above,
when $k-h\le i< k$, we have by (\ref{xiaoyuh}) that  $\|\mathbb E[\Phi_P^T(k,i+1)\Phi_P(k,i+1)|\mathcal F(i)]\|\le \rho^{*}_{h}$ a.s.; when $k-2h\le i< k-h$,   it follows from  Lemma \ref{fanshuguji} and (\ref{xiaoyuh}) that
\bann
&&~~~\|\mathbb E[\Phi_P(k,i+1)\Phi_P^T(k,i+1)|\mathcal F(i)]\|\cr
&&\le Nn\|\mathbb E[\Phi_P^T(k,i+1)\Phi_P(k,i+1)|\mathcal F(i)]\|\cr
&&=Nn\|\mathbb E[\Phi_P^T(k-h,i+1)\Phi_P^T(k,k-h+1)\Phi_P(k,k-h+1)\Phi_P(k-h,i+1)|\mathcal F(i)]\|\cr
&&=Nn\|\mathbb E[\mathbb E[\Phi_P^T(k-h,i+1)\Phi_P^T(k,k-h+1)\cr
&&~~~~\times\Phi_P(k,k-h+1)\Phi_P(k-h,i+1)|\mathcal F(k-h)]|\mathcal F(i)]\|\cr
&&=Nn\|\mathbb E[\Phi_P^T(k-h,i+1)\mathbb E[\Phi_P^T(k,k-h+1)\cr
&&~~~~\times\Phi_P(k,k-h+1)|\mathcal F(k-h)]\Phi_P(k-h,i+1)|\mathcal F(i)]\|\cr
&&\le Nn\|\mathbb E[\Phi_P^T(k-h,i+1)\cr
&&~~~~\times\|\mathbb E[\Phi_P^T(k,k-h+1)\Phi_P(k,k-h+1)|\mathcal F(k-h)]\|\Phi_P(k-h,i+1)|\mathcal F(i)]\|\cr
&&\le Nn \rho^{*}_{h}\|\mathbb E[\Phi_P^T(k-h,i+1)\Phi_P(k-h,i+1)|\mathcal F(i)]\|\le Nn(\rho^{*}_{h})^2~\mathrm{a.s.};
\eann
when $k-3h\le i< k-2h$, similar to the above, we have $\|\mathbb E[\Phi_P(k,i+1)\Phi_P^T(k,i+1)|\mathcal F(i)]\|\le Nn(\rho^{*}_{h})^3~\mathrm{a.s.}$
Hence, by  Assumptions~\textbf{A1.a} and \textbf{A1.b}, we have  $$\sup_{k\ge0}\|\mathbb E[\Phi_P(k,i+1){\mathcal H}^T(i)v(i)v^T(i){\mathcal H}(i)\Phi_P^T(k,i+1)]\|<\infty,k-3h\le i\le k,~\mathrm{a.s.}$$ Then, noting that $a(k)$ decays to zero, the second term on the right side of (\ref{s1s}) tends to zero.

We next prove that the third term  on the right side of (\ref{s1s}) tends to zero.  Let $\widetilde m_i=\lceil \frac{i}{h}\rceil.$      We have
\bnaa\label{sarf3q9qsdka}
&&~~~\|\mathbb E[\Phi_P(k,i+1)\Phi_P^T(k,i+1)|\mathcal F(i)]\|\cr
&&\le Nn\|\mathbb E[\Phi_P^T(k,i+1)\Phi_P(k,i+1)|\mathcal F(i)]\|\cr
&&=Nn\|\mathbb E[\Phi_P^T(\widetilde m_{i+1}h-1,i+1)\Phi_P^T(m_kh-1,\widetilde m_{i+1}h)\Phi_P^T(k,m_kh)\cr
&&~\times\Phi_P(k,m_kh)\Phi_P(m_kh-1,\widetilde m_{i+1}h)\Phi_P(\widetilde m_{i+1}h-1,i+1)|\mathcal F(i)]\|\cr
&&=Nn\|\mathbb E[\mathbb E[\Phi_P^T(\widetilde m_{i+1}h-1,i+1)\Phi_P^T(m_kh-1,\widetilde m_{i+1}h)\Phi_P^T(k,m_kh)\Phi_P(k,m_kh)\cr
&&~\times\Phi_P(m_kh-1,\widetilde m_{i+1}h)\Phi_P(\widetilde m_{i+1}h-1,i+1)|\mathcal F(m_kh-1)]|\mathcal F(i)]\|\cr
&&=Nn\|\mathbb E[\Phi_P^T(\widetilde m_{i+1}h-1,i+1)\Phi_P^T(m_kh-1,\widetilde m_{i+1}h)\cr
&&~\times\mathbb E[\Phi_P^T(k,m_kh)\Phi_P(k,m_kh)|\mathcal F(m_kh-1)]\cr
&&~\times\Phi_P(m_kh-1,\widetilde m_{i+1}h)\Phi_P(\widetilde m_{i+1}h-1,i+1)|\mathcal F(i)]\|\cr
&&\le Nn\rho^{*}_{h}\|\mathbb E[\Phi_P^T(\widetilde m_{i+1}h-1,i+1)\Phi_P^T(m_kh-1,\widetilde m_{i+1}h)\cr
&&~\times\Phi_P(m_kh-1,\widetilde m_{i+1}h)\Phi_P(\widetilde m_{i+1}h-1,i+1)|\mathcal F(i)]\|, \ \rm{a.s.},
\enaa
where the first inequality follows by Lemma \ref{fanshuguji}, the second equality follows by  (\ref{klkjl}) and the last inequality follows by  (\ref{xiaoyuh}).
 Similarly to  (\ref{lianchengji}) in the proof of Lemma \ref{mnbhgyls12},  we have
\bann
\|\mathbb E[\Phi_P^T(m_kh-1,\widetilde m_{i+1}h)\Phi_P(m_kh-1,\widetilde m_{i+1}h)|\mathcal F(\widetilde m_{i+1}h-1)]\|\le \prod_{s=\widetilde m_{i+1}}^{m_k-1}[1-c(s)+b^2(sh)\a],
\eann
From the above (\ref{xiaoyuh}) and~(\ref{sarf3q9qsdka}), we have
\bnaa\label{asndoajd}
&&~~~\|\mathbb E[\Phi_P(k,i+1)\Phi_P^T(k,i+1)|\mathcal F(i)]\|\cr
&&\le Nn\rho^{*}_{h}\|\mathbb E[\Phi_P^T(\widetilde m_{i+1}h-1,i+1)\Phi_P^T(m_kh-1,\widetilde m_{i+1}h)\cr
&&~\times\Phi_P(m_kh-1,\widetilde m_{i+1}h)\Phi_P(\widetilde m_{i+1}h-1,i+1)|\mathcal F(i)]\|\cr
&&=Nn\rho^{*}_{h}\|\mathbb E[\mathbb E[\Phi_P^T(\widetilde m_{i+1}h-1,i+1)\Phi_P^T(m_kh-1,\widetilde m_{i+1}h)\cr
&&~\times\Phi_P(m_kh-1,\widetilde m_{i+1}h)\Phi_P(\widetilde m_{i+1}h-1,i+1)|\mathcal F(\widetilde m_{i+1}h-1)]|\mathcal F(i)]\|\cr
&&=Nn\rho^{*}_{h}\|\mathbb E[\Phi_P^T(\widetilde m_{i+1}h-1,i+1)\cr
&&~\times\mathbb E[\Phi_P^T(m_kh-1,\widetilde m_{i+1}h)\Phi_P(m_kh-1,\widetilde m_{i+1}h)|\mathcal F(\widetilde m_{i+1}h-1)]\cr
&&~\times\Phi_P(\widetilde m_{i+1}h-1,i+1)|\mathcal F(i)]\|\cr
&&\le Nn\rho^{*}_{h}\|\mathbb E[\Phi_P^T(\widetilde m_{i+1}h-1,i+1)\|\cr
&&~\times\mathbb E[\Phi_P^T(m_kh-1,\widetilde m_{i+1}h)\Phi_P(m_kh-1,\widetilde m_{i+1}h)|\mathcal F(\widetilde m_{i+1}h-1)]\|\cr
&&~\times\Phi_P(\widetilde m_{i+1}h-1,i+1)|\mathcal F(i)]\|\cr
&&\le Nn\rho^{*}_{h}\prod_{s=\widetilde m_{i+1}}^{m_k-1}[1-c(s)+b^2(sh)\a]\|\mathbb E[\Phi_P^T(\widetilde m_{i+1}h-1,i+1)\Phi_P(\widetilde m_{i+1}h-1,i+1)|\mathcal F(i)]\|\cr
&&\le Nn(\rho^{*}_{h})^2\prod_{s=\widetilde m_{i+1}}^{m_k-1}[1-c(s)+b^2(sh)\a],~ 0\leq i \leq k-3h-1,\ \rm{a.s.},
\enaa
By (\ref{asndoajd}), the condition (b.2), Assumptions~\textbf{A1.a} and \textbf{A1.b}, it follows that
\bann
&&~~~\|\mathbb E[\Phi_P(k,i+1){\mathcal H}^T(i)v(i)v^T(i){\mathcal H}(i)\Phi_P^T(k,i+1)]\|\cr
&&=\|\mathbb E[\mathbb E[\Phi_P(k,i+1){\mathcal H}^T(i)v(i)v^T(i){\mathcal H}(i)\Phi_P^T(k,i+1)|\mathcal F(i)]]\|\cr
&&\le\|\mathbb E[\|{\mathcal H}^T(i)v(i)v^T(i){\mathcal H}(i)\|\mathbb E[\Phi_P(k,i+1)\Phi_P^T(k,i+1)|\mathcal F(i)]]\|\cr
&&\le\mathbb E[\|{\mathcal H}^T(i)v(i)v^T(i){\mathcal H}(i)\|\|\mathbb E[\Phi_P(k,i+1)\Phi_P^T(k,i+1)|\mathcal F(i)]\|]\cr
&&\le Nn (\rho^{*}_{h})^2\mathbb E[\|{\mathcal H}^T(i)v(i)v^T(i){\mathcal H}(i)\|]\prod_{s=\widetilde m_{i+1}}^{m_k-1}[1-c(s)+b^2(sh)\a]\cr
&&\le Nn\beta_v\rho_0(\rho^{*}_{h})^2\prod_{s=\widetilde m_{i+1}}^{m_k-1}[1-c(s)+b^2(sh)\a]\cr
&&\le Nn\beta_v\rho_0(\rho^{*}_{h})^2\prod_{s=\widetilde m_{i+1}}^{m_k-1}[1-\frac{1}{2}c(s)]~, m_3h-1\le i\le k-3h-1.
\eann
By the above, we have
\bnaa\label{genju28}
&&~~~\sum_{i=0}^{k-3h-1}a^2(i)\|\mathbb E[\Phi_P(k,i+1){\mathcal H}^T(i)v(i)v^T(i){\mathcal H}(i)\Phi_P^T(k,i+1)]\|\cr
&&=\sum_{i=0}^{m_3h-2}a^2(i)\|\mathbb E[\Phi_P(k,i+1){\mathcal H}^T(i)v(i)v^T(i){\mathcal H}(i)\Phi_P^T(k,i+1)]\|\cr
&&~~+\sum_{i=m_3h-1}^{k-3h-1}a^2(i)\|\mathbb E[\Phi_P(k,i+1){\mathcal H}^T(i)v(i)v^T(i){\mathcal H}(i)\Phi_P^T(k,i+1)]\|\cr
&&\le\sum_{i=0}^{m_3h-2}a^2(i)\|\mathbb E[\Phi_P(k,i+1)\Phi_P^T(k,i+1)\mathbb E[\|{\mathcal H}(i)\|^2\|v(i)\|^2\||\mathcal F(i)]]\|\cr
&&~~+ Nn\beta_v\rho_0(\rho^{*}_{h})^2\sum_{i=m_3h-1}^{k-3h-1}a^2(i)\prod_{s=\widetilde m_{i+1}}^{m_k-1}[1-\frac{1}{2}c(s)]\cr
&&\le\beta_v\rho_0\sum_{i=0}^{m_3h-2}a^2(i)\|\mathbb E[\Phi_P(k,i+1)\Phi_P^T(k,i+1)]\|\cr
&&~~+ Nn\beta_v\rho_0(\rho^{*}_{h})^2\sum_{i=m_3h-1}^{k-3h-1}a^2(i)\prod_{s=\widetilde m_{i+1}}^{m_k-1}[1-\frac{1}{2}c(s)].
\enaa
By Lemma \ref{mnbhgyls12}, we know that $\lim_{k\to\infty}\|\mathbb E[\Phi_P(k,i+1)\Phi_P^T(k,i+1)]\|=0, 0\le i\le m_3h-2$. Then,
\bnaa\label{ap9sdiqsfna}
\lim_{k\to\infty}\beta_v\rho_0\sum_{i=0}^{m_3h-2}a^2(i)\|\mathbb E[\Phi_P(k,i+1)\Phi_P^T(k,i+1)]\|=0.
\enaa

By direct calculations, it follows that
\bna\label{sdijawdaSD}
&&~~~\sum_{i=m_3h-1}^{k-3h-1}a^2(i)\prod_{s=\widetilde m_{i+1}}^{m_k-1}[1-\frac{1}{2}c(s)]\cr
&&\le\sum_{i=0}^{k}a^2(i)\prod_{s=\widetilde m_{i+1}}^{m_k-1}[1-\frac{1}{2}c(s)]\cr
&&=\sum_{i=0}^{m_kh-1}a^2(i)\prod_{s=\widetilde m_{i+1}}^{m_k-1}[1-\frac{1}{2}c(s)]+\sum_{i=m_kh}^{k}a^2(i)\prod_{s=\widetilde m_{i+1}}^{m_k-1}[1-\frac{1}{2}c(s)]\cr
&&=\sum_{i=0}^{m_k-1}\Bigg[\sum_{j=ih}^{(i+1)h-1}a^2(j)\Bigg]\prod_{s=\widetilde m_{i+1}}^{m_k-1}[1-\frac{1}{2}c(s)]+\sum_{i=m_kh}^{k}a^2(i)\prod_{s=\widetilde m_{i+1}}^{m_k-1}[1-\frac{1}{2}c(s)].
\ena
Since $a(k)$ decays to zero, it follows that
\bnaa\label{sadk9280df}
\lim_{k\to\infty}\sum_{i=m_kh}^{k}a^2(i)\prod_{s=\widetilde m_{i+1}}^{m_k-1}[1-\frac{1}{2}c(s)]=0.
\enaa
By (\ref{cmsequence}) and Condition \textbf{C1.a}, we have
$$\frac{\sum_{j=(m_k-1)h}^{m_kh-1}a^2(j)}{c(m_k-1)}\le\frac{ha^2((m_k-1)h)}{ c(m_k-1)}$$
and
$$
\lim_{k\to\infty}\frac{ha^2((m_k-1)h)}{ c(m_k-1)}=\lim_{k\to\infty}\frac{ha^2((m_k-1)h)}{b^2((m_k-1)h)}\frac{b^2((m_k-1)h)}{c(m_k-1)}=0.
$$
Then, from (\ref{cmsequence}) and  Lemma \ref{guotime}, we have
\bann
\lim_{k\to\infty}\sum_{i=0}^{m_k-1}\Bigg[\sum_{j=ih}^{(i+1)h-1}a^2(j)\Bigg]\prod_{s=\widetilde m_{i+1}}^{m_k-1}[1-\frac{1}{2}c(s+1)]=\lim_{k\to\infty}\frac{2\sum_{j=(m_k-1)h}^{m_kh-1}a^2(j)}{c(m_k-1)}=0.
\eann
By the above, (\ref{sdijawdaSD}) and (\ref{sadk9280df}), it follows that
\bnaa\label{sapdskmfs}
\lim_{k\to\infty}\sum_{i=m_3h-1}^{k-3h-1}a^2(i)\prod_{s=\widetilde m_{i+1}}^{m_k-1}[1-\frac{1}{2}c(s)]=0.
\enaa
 Then, by (\ref{genju28}), (\ref{ap9sdiqsfna}) and the above, we have
\bann\label{poepow}
\lim_{k\to\infty}\sum_{i=0}^{k-3h-1}a^2(i)\|\mathbb E[\Phi_P(k, i+1){\mathcal H}^T(i) v(i)v^T(i){\mathcal H}(i)\Phi_P^T(k, i+1)]\|=0.
\eann
Thus, the third term  on the right side of (\ref{s1s}) tends to zero.  We have
$\lim$$_{k\to\infty}$$\|$$\mathbb E$$[e(k)$$e^T(k)]$$\|$$=0$. Since~$\mathbb E$$\|e(k)\|^2$$\le$$ Nn\|$$\mathbb E[e(k)e^T(k)]\|, $  it follows that $\lim_{k\to\infty}\mathbb E\|e(k)\|^2=0$. The  algorithm~(\ref{asapp}) converges in mean square.

We next prove that the  algorithm~(\ref{asapp}) converges almost surely. By~(\ref{s1s121}), it follows that
\bann
e((m+1)h)&=&\Phi_P((m+1)h-1,mh)e(mh)\cr
&&~~+\sum_{k=mh}^{(m+1)h-1}a(k)\Phi_P((m+1)h-1,k+1)\mathcal H^T(k)v(k),m\ge0,
\eann
Taking the 2-norm and then conditional expectation w.r.t.~$\mathcal F(mh-1)$ on both sides of the above,   we have
\bann
&&~~~\mathbb E[\|e((m+1)h)\|^2|\mathcal F(mh-1)]\cr
&&=e^T(mh)\mathbb E[\Phi_P^T((m+1)h-1,mh)\Phi_P((m+1)h-1,mh)|\mathcal F(mh-1)]e(mh)\cr
&&~+\mathbb E\Bigg[\Bigg(\sum_{k=mh}^{(m+1)h-1}a(k)\Phi_P((m+1)h-1,k+1)\mathcal H^T(k)v(k)\Bigg)^T\cr
&&~\times\Bigg(\sum_{k=mh}^{(m+1)h-1}a(k)\Phi_P((m+1)h-1,k+1)\mathcal H^T(k)v(k)\Bigg)\Bigg|\mathcal F(mh-1)\Bigg]\cr
&&~+2e^T(mh)\mathbb E\Bigg[\Phi_P^T((m+1)h-1,mh)\cr
&&~\times\Bigg(\sum_{k=mh}^{(m+1)h-1}a(k)\Phi_P((m+1)h-1,k+1)\mathcal H^T(k)v(k)\Bigg)\Bigg|\mathcal F(mh-1)\Bigg].
\eann
By Lemma A.1 in~\cite{iwerw} and Assumptions \textbf{A1.a} and \textbf{A1.b}, the above can be written as
\bnaa\label{wjsd12}
&&~~~\mathbb E[\|e((m+1)h)\|^2|\mathcal F(mh-1)]\cr
&&=e^T(mh)\mathbb E[\Phi_P^T((m+1)h-1,mh)\Phi_P((m+1)h-1,mh)|\mathcal F(mh-1)]e(mh)\cr
&&~~~~+\sum_{k=mh}^{(m+1)h-1}a^2(k)\mathbb E[\|\Phi_P((m+1)h-1,k+1)\mathcal H^T(k)v(k)\|^2|\mathcal F(mh-1)].
\enaa
In the light of the condition (b.2),  Assumptions \textbf{A1.a} and \textbf{A1.b},   we know that there exists a constant~$\rho_4$ such that
\bann
\sum_{k=mh}^{(m+1)h-1}\mathbb E[\|\Phi_P((m+1)h-1,k+1)\mathcal H^T(k)v(k)\|^2|\mathcal F(mh-1)]\le \rho_4~\mathrm{a.s.},~\forall~m\ge0,
\eann
which together with (\ref{0oooopw111}) and (\ref{wjsd12}) gives
\bann
&&~~~\mathbb E[\|e((m+1)h)\|^2|\mathcal F(mh-1)]\cr
&&\le\|\mathbb E[\Phi_P^T((m+1)h-1,mh)\Phi_P((m+1)h-1,mh)|\mathcal F(mh-1)]\|\|e(mh)\|^2\cr
&&~~~~+a^2(mh)\sum_{k=mh}^{(m+1)h-1}\mathbb E[\|\Phi_P((m+1)h-1,k+1)\mathcal H^T(k)v(k)\|^2|\mathcal F(mh-1)]\cr
&&\le (1+b^2(mh)\a)\|e(mh)\|^2+a^2(mh)\rho_4~\rm{a.s.}
\eann
By  Lemma~\ref{robbins} and Condition \textbf{C1.c}, we  know that $\{e(mh),m\ge0\}$ converges almost surely, which, along with $\lim_{m\to0}\mathbb E\|e(mh)\|^2=0$  by Theorem~\ref{nodelay111}, gives  \bnaa\label{akosjdd90123i}
\lim_{m\to0}e(mh)=\textbf{0}_{Nn\times1}~\mathrm{a.s.}
\enaa
For arbitrarily small~$\epsilon>0$, by Markov inequality, we have
\bann
\mathbb P\{a(k)\|v(k)\|\ge\epsilon\}
\le\frac{a^2(k)\mathbb E\|v(k)\|^2}{\epsilon^2},~k\ge0,
\eann
which together with Assumption \textbf{A1.b}, Conditions \textbf{C1.a} and \textbf{C1.c} gives
\bann
\sum_{k=0}^\infty\mathbb P\{a(k)\|v(k)\|\ge\epsilon\}\le\frac{\sum_{k=0}^\infty a^2(k)\mathbb E\|v(k)\|^2}{\epsilon^2}\le\frac{\beta_v\sum_{k=0}^\infty a^2(k)}{\epsilon^2}<\infty.
\eann
Then by the Borel-Cantelli lemma, we have~$\mathbb P\{a(k)\|v(k)\|\ge\epsilon\ \mbox{i.o.}\}=0$, which means
\bnaa\label{ioio}
a(k)\|v(k)\|\to0,~k\to\infty~\rm{a.s.}
\enaa

By~(\ref{s1s121}), we have
\bnaa\label{sadfkse}
\|e(k)\|\le\|\Phi_P(k-1,m_kh)\|\|e(m_kh)\|+\sum_{i=m_kh}^{k-1}a(i)\|v(i)\|\|\Phi_P(k-1,i+1)\|\|\mathcal H^T(i)\|.
\enaa
By Assumption \textbf{A2.a} and noting~$0\le k-m_kh< h,$ we know that~$\sup_{k\ge0}\|\Phi_P(k-1,m_kh)\|<\infty~\mathrm{a.s.}$ and $\sup_{k\ge0}\|\Phi_P(k-1,i+1)\|\|\mathcal H^T(i)\|<\infty~\mathrm{a.s.},~m_kh\le i\le k-1$. Then by (\ref{akosjdd90123i})-(\ref{sadfkse}), we have $\lim_{k\to\infty}e(k)=\mathbf{0}_{Nn\times 1}~\mathrm{a.s.}$ The proof is completed.
\end{proof}

\begin{proof}[\textbf{Proof of Theorem~\ref{nodelay}}]
Since $\{{\mathcal G(k)},k\ge0\}\in\Gamma_1$,   $\mathbb E[\widehat {\mathcal L}_{\mathcal G(k)}|\mathcal F(k-1)] $ is positive semi-definite, which together with   $\mathbb E[\widehat {\mathcal L}_{\mathcal G(k)}|\mathcal F(mh-1)]=\mathbb E[\mathbb E[\widehat {\mathcal L}_{\mathcal G(k)}|\mathcal F(k-1)]|\mathcal F(mh-1)] $ leads to that  $\mathbb E[\widehat {\mathcal L}_{\mathcal G(k)}|\mathcal F(mh-1)]$ is positive semi-definite, $k\ge mh$. Let $c(m)=\min\{a((m+1)h),b((m+1)h)\}$. Then,  by Condition \textbf{C1.a} and the condition (c.1), we have
\bann\label{ooooooo}
&&~~~\overline{\Lambda}_m^h\cr
&&=\lambda_{\min}\Bigg[\sum_{k=mh}^{(m+1)h-1}\Bigg(b(k)\mathbb E[\widehat {\mathcal L}_{\mathcal G(k)}|\mathcal F(mh-1)]\otimes I_n+a(k)\mathbb E[{\mathcal H}^T(k){\mathcal H}(k)|\mathcal F(mh-1)]\Bigg)\Bigg]\cr
&&\ge\lambda_{\min}\Bigg[\sum_{k=mh}^{(m+1)h-1}\Bigg(b((m+1)h)\mathbb E[\widehat {\mathcal L}_{\mathcal G(k)}|\mathcal F(mh-1)]\otimes I_n\cr
&&~~+a((m+1)h)\mathbb E[{\mathcal H}^T(k){\mathcal H}(k)|\mathcal F(mh-1)]\Bigg)\Bigg]\cr
&&\ge c(m){{\Lambda_m^h}}\ge c(m)\theta.
\eann

Note that
$\sum_{m=0}^{\infty}a((m+1)h)\ge\frac{1}{h}\sum_{s=0}^\infty\sum_{i=(m+1)h}^{(m+2)h-1}a(i)=\frac{1}{h}\sum_{k=h}^\infty a(k)$.
This together with Conditions \textbf{C1.a} and \textbf{C1.b}, and $c(m)\ge\min\{a((m+1)h),a((m+1)h)/C_1 \}=\min\{1,1/C_1\}a((m+1)h)$
where $C_1\triangleq\sup_{k\ge0}\frac{a(k)}{b(k)}$, gives
\bna
\label{aspifjpazdf09}
\sum_{m=0}^\infty c(m)\geq\min\{1,1/C_1\}\sum_{m=0}^\infty a((m+1)h)\geq\frac{\min\{1,1/C_1\}}{h}\sum_{k=h}^\infty a(k)=\infty.
\ena
By Conditions \textbf{C1.a} and \textbf{C1.b}, we get
\bann\label{aspifjpazdf092}
\sup_{m\ge0}\frac{a(mh)}{c(m)}&=&\sup_{m\ge0}\frac{a(mh)}{a(mh+h)}\frac{a(mh+h)}
{c(m)}\cr
&\le&\sup_{m\ge0}\frac{a(mh)}{a(mh+h)}\frac{a(mh+h)}
{\min\{a(mh+h),\frac{1}{C_1}a(mh+h)\}}<\infty,
\eann
which together with Condition \textbf{C1.b} gives
\bnaa\label{aspifjpazdf093}
\lim_{m\to\infty}\frac{b^2(mh)}{c(m)}=\lim_{m\to\infty}\frac{b^2(mh)}{a(mh)}\frac{a(mh)}{c(m)}=0.
\enaa
Then, $c(m)$ satisfies $b^2(mh)=o(c(m))~\text{and}~\sum_{m=0}^{\infty}c(m)=\infty$. The proof is completed  by Theorem~\ref{nodelay111}.
\end{proof}

\vskip 0.2cm
\begin{proof}[\textbf{Proof of Corollary~\ref{marke}}]
By Assumption \textbf{A3} and the one-to-one correspondence among $\mathcal A_{\mathcal G(k)}$ and $ {\mathcal L}_{\mathcal G(k)}$, we know that ${\mathcal L}_{\mathcal G(k)}$ is a homogeneous and uniform ergodic Markov chain (See Definition \ref{sadkasdd90we}) with the unique stationary
distribution $\pi$. Denote the associated Laplacian matrix of $\mathcal A_l$ by ${\mathcal L}_l$ and $\widehat {\mathcal L}_l=\frac{\widehat {\mathcal L}_l+\widehat {\mathcal L}_l^T}{2}$, $l=1,2,...$
By the definition of ${{\Lambda_m^h}}$, we have
\bnaa\label{askfjak}
{{\Lambda_m^h}}&=&\lambda_{\min}\Bigg[\sum_{k=mh}^{(m+1)h-1}\mathbb E[\widehat {\mathcal L}_{\mathcal G(k)}\otimes I_n+\mathcal H^T(k)\mathcal H(k)|\mathcal F(mh-1)]\Bigg]\cr
&=&\lambda_{\min}\Bigg[\sum_{k=mh}^{(m+1)h-1}\mathbb E[\widehat {\mathcal L}_{\mathcal G(k)}\otimes I_n+\mathcal H^T(k)\mathcal H(k)|\langle\widehat {\mathcal L}_{\mathcal G(mh-1)},\mathcal H(mh-1)\rangle=S_0]\Bigg]\cr
&=&\lambda_{\min}\Bigg[\sum_{k=1}^h\sum_{l=1}^\infty (\widehat {\mathcal L}_l\otimes I_n+ \mathcal H^T_l\mathcal H_l)\mathbb P^k(S_{0},\langle\widehat {\mathcal L}_l,\mathcal H_l\rangle)\Bigg],\cr
&&~~~~~~~~~~~~~~~~~~~~~~~~~~~~~~~~~~~~\forall\ S_0\in {\mathcal S},\
\forall\ m\ge0, h\geq1.
\enaa
Noting the uniform ergodicity of $\{\widehat {\mathcal L}_{\mathcal G(k)},k\ge0\}$  and $\{\mathcal H(k),k\ge0\}$ and the uniqueness  of the stationary
distribution $\pi$, since $\sup_{l\ge1}\|{\mathcal L}_l\|<\infty$ and $  \sup_{l\ge1}\|\mathcal H_l\|<\infty$, we have
\bann
&&~~~\Bigg\|\frac{\sum_{k=1}^h\sum_{l=1}^\infty (\widehat {\mathcal L}_l\otimes I_n+ \mathcal H^T_l\mathcal H_l)\mathbb P^k(S_{0},\langle\widehat {\mathcal L}_l,\mathcal H_l\rangle)}{h}-\sum_{l=1}^\infty\pi_l(\widehat {\mathcal L}_l\otimes I_n+\mathcal H^T_l\mathcal H_l)\Bigg\|\cr
&&=\Bigg\|\frac{\sum_{k=1}^h\sum_{l=1}^\infty [(\widehat {\mathcal L}_l\otimes I_n+ \mathcal H^T_l\mathcal H_l)\mathbb P^k(S_{0},\langle\widehat {\mathcal L}_l,\mathcal H_l\rangle)-\pi_l(\widehat {\mathcal L}_l\otimes I_n+\mathcal H^T_l\mathcal H_l)]}{h}\Bigg\|\cr
&&=\Bigg\|\frac{\sum_{k=1}^h\sum_{l=1}^\infty [(\widehat {\mathcal L}_l\otimes I_n+ \mathcal H^T_l\mathcal H_l)(\mathbb P^k(S_{0},\langle\widehat {\mathcal L}_l,\mathcal H_l\rangle)-\pi_l)]}{h}\Bigg\|\cr
&&\le\sup_{l\ge1}\|\widehat {\mathcal L}_l\otimes I_n+ \mathcal H^T_l\mathcal H_l\|\frac{\sum_{k=1}^hRr^{-k}}{h}\to0,\ h\to\infty,
\eann
where constants $R$ and $r$ are positive with $r>1$.
By the definition of uniform convergence,  we know that
\bann
&&~~~\frac{1}{h}\left[\sum_{k=mh}^{(m+1)h-1}\mathbb E[\widehat {\mathcal L}_{\mathcal G(k)}\otimes I_n+\mathcal H^T(k)\mathcal H(k)|\mathcal F(mh-1)]\right]\mbox{ converges to }\cr
&&~~~~\sum_{l=1}^\infty\pi_l(\widehat {\mathcal L}_l\otimes I_n+\mathcal H^T_l\mathcal H_l)~\mbox{uniformly w.r.t. }m~\mbox{and the sample paths} ~\mathrm{a.s.}, ~\mathrm{as}~h\to\infty.\eann

By the conditions (d.1) and (d.2), it follows that $\lambda_{\min}(\sum_{l=1}^\infty\pi_l(\widehat {\mathcal L}_l\otimes I_n+\mathcal H^T_l\mathcal H_l))>0$. To see this, for
any given $x\in\mathbb{R}^{Nn}$, $x\neq\textbf{0}_{Nn\times1}$, let $x=[x_1^T,\cdots,x_N^T]^T$, $x_i\in\mathbb R^n$; (i) if $x=\textbf{1}_{N}\otimes a$, $\exists\ a\in \mathbb R^n$ and $a\not=\textbf{0}_{n\times1}$, i.e. $x_1=x_2=..=x_N=a$, then by  the condition  (d.2), we have $x^T(\sum_{l=1}^\infty\pi_l(\widehat {\mathcal L}_l\otimes I_n+\mathcal H^T_l\mathcal H_l))x=a^T[\sum_{i=1}^N\sum_{l=1}^\infty (\pi_lH_{i,l}^TH_{i,l})]a>0$; (ii) otherwise, there must be $x_i\not=x_j$, $\exists\ i\not=j$. By the condition  (d.1), we know that $\sum_{l=1}^\infty\pi_l\widehat {\mathcal L}_l$ is the Laplacian matrix of a connected graph.  Then by Lemma \ref{keyconnectivitylemmaprepareadd}, we have $x^T(\sum_{l=1}^\infty\pi_l(\widehat {\mathcal L}_l\otimes I_n+\mathcal H^T_l\mathcal H_l))x\ge x^T(\sum_{l=1}^\infty\pi_l\widehat {\mathcal L}_l\otimes I_n)x>0$. Combining (i) and (ii), we get $\lambda_{\min}(\sum_{l=1}^\infty\pi_l(\widehat {\mathcal L}_l\otimes I_n+\mathcal H^T_l\mathcal H_l))>0$.

Since the function $\lambda_{\min}(\cdot)$, whose arguments are matrices, is continuous, we know that for a given constant ${\mu}\in(0, {2}\lambda_{\min}(\sum_{l=1}^\infty\pi_l(\widehat {\mathcal L}_l\otimes I_n+\mathcal H^T_l\mathcal H_l)))$, there exists
a constant $\d>0$ such that for any given
matrix~$L$, $|\lambda_{\min}(L)-\lambda_{\min}(\sum_{l=1}^\infty\pi_l(\widehat {\mathcal L}_l\otimes I_n+\mathcal H^T_l\mathcal H_l))|\le\frac{\mu}{2}$
~provided~$\|L-\sum_{l=1}^\infty\pi_l(\widehat {\mathcal L}_l\otimes I_n+\mathcal H^T_l\mathcal H_l)\|\le\d$.
Since the convergence is uniform,
we know that there exists an  integer $h_0>0$ such that
\bann
&&~~\sup_{m\ge0}\Bigg\|\frac{1}{h}\sum_{k=mh}^{(m+1)h-1}\mathbb E[\widehat {\mathcal L}_{\mathcal G(k)}\otimes I_n+\mathcal H^T(k)\mathcal H(k)|\mathcal F(mh-1)]-\sum_{l=1}^\infty\pi_l(\widehat {\mathcal L}_l\otimes I_n+\mathcal H^T_l\mathcal H_l)\Bigg\|\cr
&&\le\d,\ h\geq h_0~\rm{a.s.},
\eann
which gives
\ban
&&~~~\sup_{m\ge0}\Bigg|\frac{1}{h}{{\Lambda_m^h}}-\lambda_{\min}\Bigg(\sum_{l=1}^\infty\pi_l(\widehat {\mathcal L}_l\otimes I_n+\mathcal H^T_l\mathcal H_l)\Bigg)\Bigg|\le\frac{\mu}{2},~h\geq h_0~\rm{a.s.}
\ean
Thus,  we arrive at
\bann
\inf_{m\ge0}{{\Lambda_m^h}}&\ge&\Bigg[\lambda_{\min}\Bigg(\sum_{l=1}^\infty\pi_l(\widehat {\mathcal L}_l\otimes I_n+\mathcal H^T_l\mathcal H_l)\Bigg)-\frac{\mu}{2}\Bigg]h\cr
&\ge&\Bigg[\lambda_{\min}\Bigg(\sum_{l=1}^\infty\pi_l(\widehat {\mathcal L}_l\otimes I_n+\mathcal H^T_l\mathcal H_l)\Bigg)-\frac{\mu}{2}\Bigg]h_0>0~\rm{a.s.}
\eann
By Theorem \ref{nodelay}, the proof is completed.
\end{proof}
  \section{proofs in Section \ref{delayexist}}
\begin{proof}[\textbf{Proof of Lemma~\ref{asa0}}]
We adopt the the mathematical induction method to prove the lemma.
By~(\ref{0fkwjeewa}) and~(\ref{ede92}),  noting that~$F(k)=I_{Nn},-d\le k\le -1$, we have
\ban\label{lmadanb}
F(0)&=&I_{Nn}-[b(0)\mathcal D_{\mathcal G(0)}\otimes I_n+a(0){\mathcal H}^T(0){\mathcal H}(0)-b(0)\sum_{q=0}^d\overline A(0,q)]\cr
&=&I_{Nn}-[b(0)\mathcal D_{\mathcal G(0)}\otimes I_n+a(0){\mathcal H}^T(0){\mathcal H}(0)-b(0)\mathcal A_{\mathcal G(0)}\otimes I_n].
\ean
Note that, under Condition \textbf{C1.d},   the set $\{\psi\in(0,1)|b(0)\le f_{C_1,\beta_a,\beta_H,N,d}(\psi)\}$ is a nonempty  and  bounded closed set by the continuity  of $ f_{C_1,\beta_a,\beta_H,N,d}(\psi)$. Hence, $\psi_1$  exists.   Then, by the definition of $\psi_1$, we have
\bnaa\label{iczmapsdxalscl}
b(0)[N\beta_a+C_1\beta_H^2+
N\beta_a[(1-\psi_1)^{-(d+1)}-1]/[(1-\psi_1)^{-1}-1]]\le{\psi_1}.
\enaa
By  the above, Assumption~\textbf{A2.b} and Condition \textbf{C1.a}, we have
\bann
\|G(0)\|&=&\|b(0)\mathcal D_{\mathcal G(0)}\otimes I_n+a(0){\mathcal H}^T(0){\mathcal H}(0)-b(0)\mathcal A_{\mathcal G(0)}\otimes I_n\|\cr
&\le& b(0)\sup_{k\ge 0}\|\mathcal D_{\mathcal G(k)}\|+a(0)\sup_{k\ge0}\|{\mathcal H}^T(k){\mathcal H}(k)\|+b(0)\sup_{k\ge 0}\|\mathcal A_{\mathcal G(k)}\|\cr
&\le& b(0)[2N\beta_a+C_1\beta_H^2]\cr
&\le& b(0)[N\beta_a+C_1\beta_H^2+
N\beta_a[(1-\psi_1)^{-(d+1)}-1]/[(1-\psi_1)^{-1}-1]]\le\psi_1~\rm{a.s.}
\eann
By the above  and  Lemma~\ref{98pkk0}, noting  $\psi_1\in(0,1)$,  it follows that   $F(0)$ is invertible \rm{a.s.} and~$\|F^{-1}(0)\|\le (1-\psi_1)^{-1}~a.s$.

Assume that~$F(k)$ is invertible~\rm{a.s.} and~$\|F^{-1}(k)\|<(1-\psi_1)^{-1}~\rm{a.s.}$  for~$k=0,1,2,\cdots$.
By  (\ref{iczmapsdxalscl}), Assumption~\textbf{A2.b} and Condition \textbf{C1.a}, we have
\bann
\|G(k+1)\|&=&\|b(k+1)\mathcal D_{\mathcal G(k+1)}\otimes I_n+a(k+1){\mathcal H}^T(k+1){\mathcal H}(k+1)\cr
&&-b(k+1)\sum_{q=0}^d\overline A(k+1,q)[\Phi_F(k,k-q+1)]^{-1}\|\cr
&\le& b(k+1)[N\beta_a+C_1\beta_H^2]+b(k+1)N\beta_a\sum_{q=0}^d(1-\psi_1)^{-q}~\rm{a.s.}\cr
&\le& b(0)[N\beta_a+C_1\beta_H^2+
N\beta_a[(1-\psi_1)^{-(d+1)}-1]/[(1-\psi_1)^{-1}-1]]\cr
&\le&\psi_1.
\eann
Then By  Lemma~\ref{98pkk0}, we know that $F(k+1)$ is invertible~\rm{a.s.} and~$\|F^{-1}(k+1)\|\le (1-\psi_1)^{-1}~\rm{a.s.}$ By the mathematical induction, the proof is completed.
\end{proof}
Before proving  Theorem~\ref{e4t5634332}, we need the following lemma.
\begin{lemma}\label{erhi32}
\rm{If Assumption~\textbf{A2.b}, Conditions \textbf{C1.a} and \textbf{C1.d} hold,   and there exist a positive integer~$h$ and a positive sequence~$\{  c(m),m\ge0\}$ such that  $ {\widetilde\Lambda_m^h}\ge   c(m)~\rm{a.s.}$ with $c(m)$  satisfying
\bnaa\label{cmsequenceC1}
b^2(mh)=o(c(m))~\text{and}~\sum_{m=0}^{\infty}c(m)=\infty,
\enaa  then~$$
\lim_{k\to\infty}\big\|\mathbb E(\Phi_F(k,0)\Phi_F^T(k,0))\big\|=0.
$$}
\end{lemma}
\begin{proof}
Since Assumption~\textbf{A2.b}, Conditions \textbf{C1.a} and \textbf{C1.d} hold,  Lemma \ref{asa0} holds. Hence, $F(k)$ is invertible \rm{a.s.}, and (\ref{ede92}) follows.

 Similarly to (\ref{lesisyifdiosjd})$-$(\ref{a098okmsadsasd}) in the proof of Lemma~\ref{mnbhgyls12}, there exists a  positive integer $m_1'$ such that
\bnaa\label{a098okmsadsasd11}
&&~~~\|\mathbb E[\Phi_F((m+1)h-1,mh)\Phi_F^T((m+1)h-1,mh)|\mathcal F(mh-1)]\|\cr
~&&=1-\lambda_{\min}\Bigg(\sum_{k=mh}^{(m+1)h-1}\mathbb E[G(k)+G^T(k)|\mathcal F(mh-1)]\Bigg)\cr
&&~+\left\|\mathbb E[\overline M_2(m)+\cdots+\overline M_{2h}(m)|\mathcal F(mh-1)]\right\|,~\forall~m\ge m_1'~\rm{a.s.}
\enaa
Here, the definitions of $\overline M_i(m),i=2,\cdots,2h$ are similar to (\ref{1712323}).

By  (\ref{gkfor}), (\ref{lam})  and  $ {\widetilde\Lambda_m^h}\ge  c(m)~\rm{a.s.}$, we have
\bnaa\label{a098okssams}
&&~~~1-\lambda_{\min}\Bigg(\sum_{k=mh}^{(m+1)h-1}\mathbb E[G(k)+G^T(k)|\mathcal F(mh-1)]\Bigg)\cr
~&&=1-\lambda_{\min}\Bigg(\sum_{k=mh}^{(m+1)h-1}\mathbb E\Bigg[2b(k){\mathcal D}_{\mathcal G(k)}\otimes I_n+2a(k){\mathcal H}^T(k){\mathcal H}(k)\cr
&&~-b(j)\sum_{q=0}^d
[\overline A(k,q)[\Phi_F(k-1,k-q)]^{-1}+(\overline A(k,q)[\Phi_F(k-1,k-q)]^{-1})^T]\Big|\mathcal F(mh-1)\Bigg]\Bigg)\cr
~&&=1-\lambda_{\min}\Bigg(\sum_{k=mh}^{(m+1)h-1}\mathbb E\Bigg[2b(k)\widehat{\mathcal L}_{\mathcal G(k)}\otimes I_n+2a(k){\mathcal H}^T(k){\mathcal H}(k)\cr
&&~-b(k)\sum_{q=0}^d
[\overline A(k,q)[[\Phi_F(k-1,k-q)]^{-1}-I_{Nn}]\cr
&&~~~~~~~+(\overline A(k,q)[[\Phi_F(k-1,k-q)]^{-1}-I_{Nn}])^T]\Big|\mathcal F(mh-1)\Bigg]\Bigg)\cr
~&&=1-2{\widetilde\Lambda_m^h} \le1-  c(m)~\rm{a.s.}
\enaa

From (\ref{gkfor}), Assumption~\textbf{A2.b}, Condition \textbf{C1.a} and Lemma~\ref{asa0}, we have
\bann\label{opwjxsasa}
\|G(k)\|&\le& b(k)\|\mathcal D_{\mathcal G(k)}\otimes I_n\|+a(k)\|{\mathcal H}^T(k){\mathcal H}(k)\|+b(k)\|\sum_{q=0}^d\overline A(k,q)[\Phi_F(k-1,k-q)]^{-1}\|\cr
&\le & b(k)\Bigg(N\beta_a+C_1\beta_H^2+N\beta_a\frac{1-(1-\psi_1)^{-(d+1)}}{1-(1-\psi_1)^{-1}}\Bigg)~\mathrm{a.s.},~ k\ge 0.
\eann
By the above  and the definition of~$\overline M_i(m)$, $i=2,\cdots,2h$, we have
\ban\label{0qjs}
&&\|\overline M_i(m)\|\le b^2(mh)\mathbb C_{2h}^i\Bigg(N\beta_a+C_1\beta_H^2+N\beta_a\frac{1-(1-\psi_1)^{-(d+1)}}{1-(1-\psi_1)^{-1}}\Bigg)^i
~~\rm{a.s.},
\ean
where~$\mathbb C_m^p$ represent  the combinatorial number of choosing $p$ elements from $m$ elements.
Hence,
\bna\label{00ksapw}
&&~~~\|\mathbb E[\overline M_2(m)+\cdots+\overline M_{2h}(m)|\mathcal F(mh-1)]\|\cr
~&&\le b^2(mh)\sum_{i=2}^{2h}\mathbb C_{2h}^i\Bigg(N\beta_a+C_1\beta_H^2+N\beta_a\frac{1-(1-\psi_1)^{-(d+1)}}{1-(1-\psi_1)^{-1}}\Bigg)^i=b^2(mh)\gamma~\rm{a.s.},
\ena
where
\bann\label{nmhd}
\gamma&=&\Bigg(\Bigg(N\beta_a+C_1\beta_H^2+N\beta_a\frac{1-(1-\psi_1)^{-(d+1)}}{1-(1-\psi_1)^{-1}}\Bigg)+1\Bigg)^{2h}\cr
&&-1-2h\Bigg(N\beta_a+C_1\beta_H^2+N\beta_a\frac{1-(1-\psi_1)^{-(d+1)}}{1-(1-\psi_1)^{-1}}\Bigg).
\eann
By (\ref{a098okmsadsasd11}), (\ref{a098okssams}) and (\ref{00ksapw}), we have
\bna\label{0oooopw}
&&~~~\|\mathbb E[\Phi_F((m+1)h-1,mh)\Phi_F^T((m+1)h-1,mh)|\mathcal F(mh-1)]\|\cr
~&&\le1- c(m)+b^2(mh)\gamma~\mathrm{a.s.},~m\ge m_1'.
\ena
By (\ref{ede92}) and Assumption \textbf{A2.b}, we  know that there exists a positive constant~$\overline\eta$ such that
\bnaa\label{0asdq3}
\|F(k)\|\le\overline\eta~\mathrm{a.s.},~k\ge 0.
\enaa
Denote~$m_k=\lfloor \frac{k}{h} \rfloor$. By (\ref{0asdq3})  and~Lemma \ref{fanshuguji}, we have
\bna\label{wang1}
&&~~~\|\mathbb E[\Phi_F(k,0)\Phi_F^T(k,0)]\|\cr
&&\le Nn\|\mathbb E[\Phi_F^T(k,0)\Phi_F(k,0)]\|\cr
&&=Nn\|\mathbb E[\Phi_F^T(m_kh-1,0)\Phi_F^T(k,m_kh)\Phi_F(k,m_kh)\Phi_F(m_kh-1,0) ]\|\cr
&&\le Nn\|\mathbb E[\Phi_F^T(m_kh-1,0)\|\Phi_F(k,m_kh)\|^2\Phi_F(m_kh-1,0) ]\|\cr
&&\le {\overline\eta}^{2h}Nn\|\mathbb E[\Phi_F^T(m_kh-1,0)\Phi_F(m_kh-1,0) ]\|\cr
&&={\overline\eta}^{2h}Nn\|\mathbb E[\Phi_F^T(m_1'h-1,0)\Phi_F^T(m_kh-1,m_1'h)\Phi_F(m_kh-1,m_1'h)\Phi_F(m_1'h-1,0)]\|\cr
&&\le{\overline\eta}^{2h}Nn\|\mathbb E[\|\Phi_F(m_1'h-1,0)\|^2\Phi_F^T(m_kh-1,m_1'h)\Phi_F(m_kh-1,m_1'h)]\|\cr
&&\le{\overline\eta}^{2(h+m_1'h)}Nn\|\mathbb E[\Phi_F^T(m_kh-1,m_1'h)\Phi_F(m_kh-1,m_1'h)]\|~\rm{a.s.}
\enaa
From the properties of the conditional expectation and~(\ref{0oooopw}), it follows that
\bnaa\label{iou1}
&&~~~\|\mathbb E[\Phi_F^T(m_kh-1,m_1'h)\Phi_F(m_kh-1,m_1'h)]\|\cr
~&&=\|\mathbb E[\Phi_F^T((m_k-1)h-1,m_1'h)\Phi_F^T(m_kh-1,(m_k-1)h)\Phi_F(m_kh-1,(m_k-1)h)\cr
&&~~~~~~~\times \Phi_F((m_k-1)h-1,m_1'h)]\|\cr
~&&=\|\mathbb E[\mathbb E[\Phi_F^T((m_k-1)h-1,m_1'h)\Phi_F^T(m_kh-1,(m_k-1)h)\Phi_F(m_kh-1,(m_k-1)h)\cr
&&~~~~~~~\times \Phi_F((m_k-1)h-1,m_1'h)|\mathcal F((m_k-1)h-1)]]\|\cr
~&&\le\|\mathbb E[\Phi_F^T((m_k-1)h-1,m_1'h)\cr
&&~~~~~~~\times\|\mathbb E[\Phi_F^T(m_kh-1,(m_k-1)h)\Phi_F(m_kh-1,(m_k-1)h)|\mathcal F((m_k-1)h-1)]\|\cr
&&~~~~~~~\times \Phi_F((m_k-1)h-1,m_1'h)]\|\cr
~&&\le[1-  c(m_k-1)+b^2((m_k-1)h)\gamma]\cr
&&~~~~~~~~\times\|\mathbb E[\Phi_F^T((m_k-1)h-1,m_1'h)\Phi_F((m_k-1)h-1,m_1'h)]\|\cr
~&&\le\prod_{s=m_1'}^{m_k-1}[1-  c(s)+b^2(sh)\gamma]~\rm{a.s.}
\enaa
Combining~(\ref{wang1}) and~(\ref{iou1}) implies
\ban\label{0ooyuugw}
\|\mathbb E[\Phi_F(k,0)\Phi_F^T(k,0)]\|\le Nn{\overline\eta}^{2(h+m_1'h)}\prod_{s=m_1'}^{m_k-1}[1-  c(s)+b^2(sh)\gamma]~\rm{a.s.}
\eann
Similarly to (\ref{limit1})$-$(\ref{lesisyu42}) in the proof of Lemma~\ref{mnbhgyls12}, by  Condition~\textbf{C1.a}, (\ref{cmsequenceC1}) and the above, we have
$
\lim_{k\to\infty}\|\mathbb E[\Phi_F(k,0)\Phi_F^T(k,0)]\| =0.
$
The proof is completed.
\end{proof}

\vskip 0.2cm

\begin{proof}[\textbf{Proof of Theorem~\ref{e4t5634332}}]
By the conditions of the theorem, it follows that Lemmas \ref{asa0} and \ref{erhi32} hold.

 Denote the following block matrices: $\overline r(k)=[r^T(k),g^T(k),\cdots,g^T(k-d+1)]^T$, $\widehat I=[\textbf{0}_{Nn\times Nn},\widetilde I]^T$ and $\widetilde I=[I_{Nn},\textbf{0}_{Nn\times Nn},\cdots,\textbf{0}_{Nn\times Nn}],$ where~$\widehat I$ and $\widetilde I$ are the $Nn(d+1)$ dimensional column block matrix and $Nnd$ dimensional row block matrix with each block being the~$Nn$ dimensional matrix, respectively. Denote
 \ban
T(k)=
\left(\begin{array}{cc}
    F(k)&\widetilde I\\
     \textbf{0}_{Nnd\times Nn}& C(k)\\
  \end{array}\right),
  \ean
  which gives
\ban\label{o0kss}
\Phi_T(k,0)&=\left(
           \begin{array}{cc}
             \Phi_F(k,0) &  \sum_{i=0}^k\Phi_F(k,i+1)\widetilde I\Phi_C(i-1,0)   \\
             \textbf{0}_{Nnd\times Nn} & \Phi_C(k,0) \\
           \end{array}
         \right).
\ean
Denote
\bna\label{er2343eas}
C(k)=\left(
                \begin{array}{cccc}
                  C_1(k+1) & C_2(k+1) &\cdots&C_d(k+1)\\
                  I_{Nn} & \textbf{0}_{Nn\times Nn}& & \\
                   &\ddots &\ddots &  \\
                  &  & I_{Nn}& \textbf{0}_{Nnd\times Nn} \\
                \end{array}
              \right).
\ena

By the state augmentation approach and~(\ref{12129s}), we have
\ban\label{er234s45}
\overline r(k+1)&=&T(k)\overline r(k)+a(k+1)\widehat I{\mathcal H}^T(k+1)v(k+1)\cr
 &=&\Phi_T(k,0)\overline r(0)+\sum_{i=1}^{k+1}a(i)\Phi_T(k,i)\widehat I{\mathcal H}^T(i)v(i),~k\ge 0.
\ean
Premultiplying the $Nn(d+1)$ dimensional row block  matrix $\overline I\triangleq[I_{Nn}, \textbf{0}_{Nn\times Nn},\cdots, \textbf{0}_{Nn\times Nn}]$ on both sides of the above  gives
\bann\label{er234s}
r(k+1) =\overline I\Phi_T(k,0)\overline r(0)+\sum_{i=1}^{k+1}a(i)\overline I\Phi_T(k,i)\widehat I{\mathcal H}^T(i)v(i),
\eann
which  leads to
\bna\label{8ndi4s}
&&~~~\mathbb E[r(k+1) r^T(k+1)]\cr
&&=\mathbb E[\overline I\Phi_T(k,0)\overline r(0)\overline r^T(0)\Phi_T^T(k,0)\overline I^T]\cr
&&~~~+\mathbb E\Bigg[\overline I\Phi_T(k,0)\overline r(0)\Big(\sum_{i=1}^{k+1}a(i)v^T(i){\mathcal H}(i)\widehat I^T \Phi_T^T(k,i)\overline I^T\Big)\Bigg]\cr
&&~~~+\mathbb E\Big[\Big(\sum_{i=1}^{k+1}a(i)\overline I\Phi_T(k,i)\widehat I{\mathcal H}^T(i)v(i)\Big)\overline r^T(0)\Phi_T^T(k,0)\overline I^T\Big]\cr
&&~~~+\mathbb E\Bigg[\Big[\sum_{i=1}^{k+1}a(i)\overline I\Phi_T(k,i)\widehat I{\mathcal H}^T(i)v(i)\Big]\Big[\sum_{i=1}^{k+1}a(i)[\overline I\Phi_T(k,i)\widehat I{\mathcal H}^T(i)v(i)]^T\Big]\Bigg].
\ena
By  Assumptions~\textbf{A1.a} and \textbf{A1.b} ,
 we   know that the second  and  third terms on the right side of the above  are both equal to zero.

 By (\ref{jiaochaxiang}), we have
\bann
&&~~~~~\mathbb E\Bigg[\Big[\sum_{i=1}^{k+1}a(i)\overline I\Phi_T(k,i)\widehat I{\mathcal H}^T(i)v(i)\Big]\Big[\sum_{i=1}^{k+1}a(i)[\overline I\Phi_T(k,i)\widehat I{\mathcal H}^T(i)v(i)]^T\Big]\Bigg]\cr
&&=\sum_{i=1}^{k+1}a^2(i)\mathbb E[\overline I\Phi_T(k,i)\widehat I{\mathcal H}^T(i)v(i)v^T(i){\mathcal H}(i)\widehat I^T\Phi_T^T(k,i)\overline I^T].
\eann
Substituting the above   into (\ref{8ndi4s}) and taking the 2-norm on both sides of (\ref{8ndi4s}), from Assumptions~\textbf{A1.a}, \textbf{A1.b} and~\textbf{A2.b}, it follows that
\bna\label{asa8ndiassasas4s}
&&~~~\|\mathbb E[r(k+1) r^T(k+1)]\|~~~~~~~~~~~~~~~~~~~~~~~~~~~~~~~~\cr
&&\le r_0\|\mathbb E[\overline I\Phi_T(k,0)\Phi_T^T(k,0)\overline I^T]\|\cr
&&~~~~~+\Bigg\|\sum_{i=1}^{k+1}a^2(i)\mathbb E[\overline I\Phi_T(k,i)\widehat I{\mathcal H}^T(i)v(i)v^T(i){\mathcal H}(i)\widehat I^T\Phi_T^T(k,i)\overline I^T]\Bigg\|\cr
&&= r_0\|\mathbb E[\overline I\Phi_T(k,0)\Phi_T^T(k,0)\overline I^T]\|\cr
&&~~~~~+\Bigg\|\sum_{i=1}^{k+1}a^2(i)\mathbb E[\overline I\Phi_T(k,i)\widehat I{\mathcal H}^T(i)\mathbb E(v(i)v^T(i)){\mathcal H}(i)\widehat I^T\Phi_T^T(k,i)\overline I^T]\Bigg\|\cr
&&\le r_0\|\mathbb E[\overline I\Phi_T(k,0)\Phi_T^T(k,0)\overline I^T]\|\cr
&&~~~~+\sup_{k\ge0}\|\mathbb E[v(k)v^T(k)]\|\Bigg\|\sum_{i=1}^{k+1}a^2(i)\mathbb E[\overline I\Phi_T(k,i)\widehat I{\mathcal H}^T(i){\mathcal H}(i)\widehat I^T\Phi_T^T(k,i)\overline I^T]\Bigg\|\cr
&&\le r_0\|\mathbb E[\overline I\Phi_T(k,0)\Phi_T^T(k,0)\overline I^T]\|\cr
&&~~~~~~~~+\beta_H\sup_{k\ge0}\|\mathbb E[v(k)v^T(k)]\|\Bigg\|\sum_{i=1}^{k+1}a^2(i)\mathbb E[\overline I\Phi_T(k,i)\widehat I\widehat I^T\Phi_T^T(k,i)\overline I^T]\Bigg\|\cr
&&\le r_0\|\mathbb E[\overline I\Phi_T(k,0)\Phi_T^T(k,0)\overline I^T]\|+\beta_H^2\beta_v\sum_{i=1}^{k+1}a^2(i)\|\mathbb E[\overline I\Phi_T(k,i)\Phi_T^T(k,i)\overline I^T]\|,
\ena
where~$r_0\triangleq\|\overline r(0)\overline r^T(0)\|$.
By the definitions of~$\Phi_T(k,0)$ and~$\overline I$, we have~$$\overline I\Phi_T(k,0)=\Big(\Phi_F(k,0) \sum_{i=0}^k\Phi_F(k,i+1)\widetilde I\Phi_C(i-1,0)\Big).$$
Substituting the above  into~(\ref{asa8ndiassasas4s}) gives
\bna\label{asa8ndi4s}
&&~~~\|\mathbb E[r(k+1) r^T(k+1)]\|\cr
&&\le r_0\|\mathbb E[\Phi_F(k,0)\Phi_F^T(k,0)]\|+\beta_H^2\beta_v\sum_{i=1}^{k+1}a^2(i)\|\mathbb E[\Phi_F(k,i)\Phi_F^T(k,i)]\|\cr
&&~~~+r_0\Big\|\mathbb E\Big[\Big\{\sum_{i=0}^k\Phi_F(k,i+1)\widetilde I\Phi_C(i-1,0)\Big\}\Big\{\sum_{i=0}^k\Phi_C^T(i-1,0)\widetilde I^T\Phi_F^T(k,i+1)\Big\}\Big]\Big\|\cr
&&~~~+\beta_H^2\beta_v\sum_{i=1}^{k+1}a^2(i)\Big\|\mathbb E\Big[\Big\{\sum_{j=i}^k\Phi_F(k,j+1)\widetilde I\Phi_C(j-1,i)\Big\}\cr
&&~~~\times\Big\{\sum_{j=i}^k\Phi_F(k,j+1)\widetilde I\Phi_C(j-1,i)\Big\}^T\Big]\Big\|.
\ena
By Lemma~\ref{erhi32},  we  know that the first term on the right side of the above  converges to zero.

Denote~$\widetilde m_i=\lceil \frac{i}{h}\rceil$. By~(\ref{0asdq3}) and noting the definition of $m_k$ defined in the proof of Lemma~\ref{erhi32}, we have
\bann\label{yin4}
&&~~~\sum_{i=1}^{k-3h}a^2(i)\|\mathbb E[\Phi_F(k,i)\Phi_F^T(k,i)]\|\cr
~&&= \sum_{i=0}^{k-3h-1}a^2(i+1)\|\mathbb E[\Phi_F(k,i+1)\Phi_F^T(k,i+1)]\|\cr
~&&= \sum_{i=0}^{k-3h-1}a^2(i+1)\|\mathbb E[\Phi_F(k,m_kh)\Phi_F(m_kh-1,\widetilde m_{i+1}h)\Phi_F(\widetilde m_{i+1}h-1,i+1)\cr
&&~\times\Phi_F^T(\widetilde m_{i+1}h-1,i+1)\Phi_F^T(m_kh-1,\widetilde m_{i+1}h)\Phi_F^T(k,m_kh)]\|\cr
~&&\le {\overline\eta}^{2h}\sum_{i=0}^{k-3h-1}a^2(i+1)\|\mathbb E[\Phi_F(k,m_kh)\Phi_F(m_kh-1,\widetilde m_{i+1}h)\Phi_F^T(m_kh-1,\widetilde m_{i+1}h)\Phi_F^T(k,m_kh)]\|\cr
~&&\le {\overline\eta}^{4h}\sum_{i=0}^{k-3h-1}a^2(i+1)\|\mathbb E[\Phi_F(m_kh-1,\widetilde m_{i+1}h)\Phi_F^T(m_kh-1,\widetilde m_{i+1}h)]\|,
\eann
which together with Lemma \ref{fanshuguji} and (\ref{iou1}) leads to
\bann
&&~~~\sum_{i=1}^{k+1}a^2(i)\|\mathbb E[\Phi_F(k,i)\Phi_F^T(k,i)]\|\cr
&&\le{\overline\eta}^{4h}\sum_{i=0}^{k-3h-1}a^2(i+1)\|\mathbb E[\Phi_F(m_kh-1,\widetilde m_{i+1}h)\Phi_F^T(m_kh-1,\widetilde m_{i+1}h)]\|\cr
&&~ +\sum_{i=k-3h}^{k}a^2(i+1)\|\mathbb E[\Phi_F(k,i+1)\Phi_F^T(k,i+1)]\|\cr
~&&\le Nn{\overline\eta}^{4h}\sum_{i=0}^{k-3h-1}a^2(i+1)\|\mathbb E[\Phi_F^T(m_kh-1,\widetilde m_{i+1}h)\Phi_F(m_kh-1,\widetilde m_{i+1}h)]\|\cr
&&~ +\sum_{i=k-3h}^{k}a^2(i+1)\|\mathbb E[\Phi_F(k,i+1)\Phi_F^T(k,i+1)]\|\cr
~&&\le Nn{\overline\eta}^{4h}\sum_{i=0}^{k-3h-1}a^2(i+1)\prod_{s=\widetilde m_{i+1}}^{m_k-1}[1-c(s)+b^2(sh)\gamma]\cr
&&~ +\sum_{i=k-3h}^{k}a^2(i+1)\|\mathbb E[\Phi_F(k,i+1)\Phi_F^T(k,i+1)]\|.
\eann
Similarly to (\ref{sdijawdaSD})$-$(\ref{sapdskmfs}) in the proof of Theorem~\ref{nodelay111}, we have
\bnaa\label{quyusas}
\lim_{k\to\infty}\sum_{i=1}^{k+1}a^2(i)\|\mathbb E[\Phi_F(k,i)\Phi_F^T(k,i)]\|=0.
\enaa
Hence, the second term on the right side of (\ref{asa8ndi4s}) converges to zero.

From~(\ref{121esd29s}) and~(\ref{ede92}), we have
\bann\label{cik}
C_i(k)=-b(k)\sum_{q=i}^d\overline A(k,q)[\Phi_F(k-1,k-q)]^{-1},~1\le i\le d.
\eann
By  Assumption~\textbf{A2.b} and Condition \textbf{C1.a}, then there exist~$\epsilon\in(0,\frac{1-\psi_1}{\sqrt{Nnd}})$, where $\psi_1$ is defined in Lemma \ref{asa0} and a positive integer~$k(\epsilon)$, such that for~$\forall~k\ge k(\epsilon)$, $\|C_i(k)\|_\infty\le\frac{\epsilon(\epsilon-1)}{\epsilon-\epsilon^{1-d}}~\mathrm{a.s.},1\le i\le d$, where~$\|\cdot\|_\infty$ represents the infinite norm of a matrix. If~$d>1$, denote~$Y=diag\{I_{Nn},\epsilon I_{Nn},\epsilon^2I_{Nn},\cdots,\epsilon^{d-1}I_{Nn}\}$; if~$d=1$, denote~$Y=I_{Nn}$, which together with  (\ref{er2343eas}) leads to
\bann
YC(k)Y^{-1}=\left(
                \begin{array}{cccc}
                  C_1(k+1) &\epsilon^{-1} C_2(k+1) &\cdots&\epsilon^{1-d}C_d(k+1)\\
                  \epsilon I_{Nn} & \textbf{0}_{Nn\times Nn} & & \\
                   &\ddots &\ddots&  \\
                  &  & \epsilon I_{Nn}& \textbf{0}_{Nn\times Nn}  \\
                \end{array}
              \right).
\eann
Then, it follows that
\bann
\|YC(k)Y^{-1}\|_\infty \le\max\Big\{\sum_{i=1}^d\epsilon^{1-i}\|C_i(k+1)\|_\infty,\epsilon \Big\}\le \max\Big\{\frac{\epsilon(\epsilon-1)}{\epsilon-\epsilon^{1-d}}\frac{\epsilon-\epsilon^{1-d}}{\epsilon-1},\epsilon \Big\}=\epsilon~\rm{a.s.}
\eann
From the relation between  infinite norm  and  2-norm of a matrix, we have
\bnaa\label{ccf}
\|YC(k)Y^{-1}\|\le\sqrt{Nnd}\|YC(k)Y^{-1}\|_\infty\le\epsilon\sqrt{Nnd}<1-\psi_1~\rm{a.s.}
\enaa
Noting that~$F(k)$ is invertible~\rm{a.s.},  we have
\bna\label{0ozsdpw}
&&~~~~~\Big\|\mathbb E\Big[\Big\{\sum_{i=0}^k\Phi_F(k,i+1)\widetilde I\Phi_C(i-1,0)\Big\}\Big\{\sum_{i=0}^k\Phi_F(k,i+1)\widetilde I\Phi_C(i-1,0)\Big\}^T\Big]\Big\|\cr
&&\le\sum_{0\le i,j\le k}\|\mathbb E[\Phi_F(k,i+1)\widetilde I\Phi_C(i-1,0)\Phi_C^T(j-1,0)\widetilde I^T\Phi_F^T(k,j+1)]\|\cr
&&\le\sum_{0\le i,j\le k}\|\mathbb E[\Phi_F(k,0)[\Phi_F(i,0)]^{-1}\widetilde I\Phi_C(i-1,0)\Phi_C^T(j-1,0)\widetilde I^T[\Phi_F(j,0)]^{-T}\Phi_F^T(k,0)]\|\cr
&&\le\sum_{0\le i,j\le k}\|\mathbb E[\Phi_F(k,0)\|[\Phi_F(i,0)]^{-1}\|\|\widetilde I\Phi_C(i-1,0)\Phi_C^T(j-1,0)\widetilde I^T\|\cr
&&~\times\|[\Phi_F(j,0)]^{-T}\|\Phi_F^T(k,0)]\|.
\enaa
By Lemma~\ref{asa0}, it follows that
\bnaa\label{de78}
&&\|[\Phi_F(i,0)]^{-1}\|\le(1-\psi_1)^{-(i+1)}~\mathrm{and}~\|[\Phi_F(j,0)]^{-T}\|\le(1-\psi_1)^{-(j+1)}~\rm{a.s.}
\enaa
From~(\ref{ccf}), we obtain
\bnaa\label{84565}
\|\widetilde I\Phi_C(i-1,0)\Phi_C^T(j-1,0)\widetilde I^T\|
&\le&\|\Phi_C(i-1,0)\|\|\Phi_C(j-1,0)\|\cr
&=&\|Y^{-1}\Phi_{YCY^{-1}}(i-1,0)Y\|\|Y^{-1}\Phi_{YCY^{-1}}(j-1,0)Y\|\cr
&\le&(\epsilon\sqrt{Nnd})^{i+j-2}~\rm{a.s.},
\enaa
which combining (\ref{0ozsdpw}) and~(\ref{de78}) gives
\bann\label{41}
&&~~~~~\Big\|\mathbb E\Big[\Big\{\sum_{i=0}^k\Phi_F(k,i+1)\widetilde I\Phi_C(i-1,0)\Big\}\Big\{\sum_{i=0}^k\Phi_F(k,i+1)\widetilde I\Phi_C(i-1,0)\Big\}^T\Big]\Big\|\cr
&&\le (1-\psi_1)^{-2}\|\mathbb E[\Phi_F(k,0)\Phi_F^T(k,0)]\|\sum_{0\le i,j\le k}((1-\psi_1)^{-1}\epsilon\sqrt{Nnd})^{i+j}~\rm{a.s.}
\eann
Noting that $(1-\psi_1)^{-1}\epsilon\sqrt{Nnd}<1$, we have~$\sum_{0\le i,j< \infty}((1-\psi_1)^{-1}\epsilon\sqrt{Nnd})^{i+j}<\infty.$
Hence, by Lemma~\ref{erhi32}, it follows that
\bann\label{dl41}
\lim_{k\to\infty}\Big\|\mathbb E\Big[\Big\{\sum_{i=0}^k\Phi_F(k,i+1)\widetilde I\Phi_C(i-1,0)\Big\}\Big\{\sum_{i=0}^k\Phi_F(k,i+1)\widetilde I\Phi_C(i-1,0)\Big\}^T\Big]\Big\|=0.\eann
Thus, the third term on the right side of~(\ref{asa8ndi4s}) converges to zero.

By~(\ref{de78})-(\ref{84565}) and similarly to~(\ref{0ozsdpw}), it follows that
\bann\label{55}
&&~~~\sum_{i=1}^{k+1}a^2(i)\Big\|\mathbb E\Big[\Big\{\sum_{j=i}^k\Phi_F(k,j+1)\widetilde I\Phi_C(j-1,i)\Big\}\Big\{\sum_{j=i}^k\Phi_F(k,j+1)\widetilde I\Phi_C(j-1,i)\Big\}^T\Big]\Big\|\cr
&&=\sum_{i=1}^{k+1}a^2(i)\Big\|\sum_{i\le j_1,j_2\le k}\mathbb E [\Phi_F(k,j_1+1)\widetilde I\Phi_C(j_1-1,i)\Phi_C^T(j_2-1,i)\widetilde I ^T\Phi_F^T(k,j_2+1)]\Big\|\cr
&&=\sum_{i=1}^{k+1}a^2(i)\Big\|\sum_{i\le j_1,j_2\le k}\mathbb E [\Phi_F(k,i)(\Phi_F(j_1,i))^{-1}\widetilde I\Phi_C(j_1-1,i)\cr
&&~\times\Phi_C^T(j_2-1,i)\widetilde I ^T(\Phi_F^T(j_2,i))^{-1}\Phi_F^T(k,i)]\Big\|\cr
&&\le\sum_{i=1}^{k+1}a^2(i)\Big\|\sum_{i\le j_1,j_2\le k}\mathbb E [\Phi_F(k,i)\|(\Phi_F(j_1,i))^{-1}\widetilde I\Phi_C(j_1-1,i)\cr
&&~\times\Phi_C^T(j_2-1,i)\widetilde I ^T(\Phi_F^T(j_2,i))^{-1}\|\Phi_F^T(k,i)]\Big\|\cr
&&\le\sum_{i=1}^{k+1}a^2(i)\|\mathbb E [\Phi_F(k,i)\Phi_F^T(k,i)]\|\sum_{i\le j_1,j_2\le k}(1-\psi_1)^{-(j_1+j_2-2i+6)}(\epsilon\sqrt{Nnd})^{j_1+j_2-2i}~\rm{a.s.}\cr
&&\le(1-\psi_1)^{-6}\sum_{i=1}^{k+1}a^2(i)\|\mathbb E [\Phi_F(k,i)\Phi_F^T(k,i)]\|\sum_{i\le j_1,j_2\le k}((1-\psi_1)^{-1}\epsilon\sqrt{Nnd})^{(j_1+j_2-2i)}\cr
&&=(1-\psi_1)^{-6}\sum_{i=1}^{k+1}a^2(i)\|\mathbb E [\Phi_F(k,i)\Phi_F^T(k,i)]\|\frac{1-((1-\psi_1)^{-1}\epsilon)^{2k-2i+1}}{1-(1-\psi_1)^{-1}\epsilon\sqrt{Nnd}}\cr
&&\le\frac{(1-\psi_1)^{-6}}{1-(1-\psi_1)^{-1}\epsilon\sqrt{Nnd}}\sum_{i=1}^{k+1}a^2(i)\|\mathbb E [\Phi_F(k,i)\Phi_F^T(k,i)]\|~\rm{a.s.}~~~~~~~~~~~~~~~~~~~~~~~~~~~~~~~~~~~~~~
\eann
In the light of~(\ref{quyusas}),  the above converges to zero.

So far, we have proved that all the four terms on the right side of   (\ref{asa8ndi4s}) converge to zero. Thus, we have~$\lim_{k\to\infty}\|\mathbb E(r(k+1) r^T(k+1))\|=0$,
which, along with the facts that $\mathbb E \|r(k)\|^2=\mathbb E[\mathrm{Tr}(r(k)r^T(k))]=\mathrm{Tr}[\mathbb E(r(k)r^T(k))]$ and
$r(k)$ is equivalent to~$e(k)$, gives $\lim_{k\to\infty}\mathbb E\|e(k)\|^2=0$. The proof is completed.
\end{proof}

\vskip 0.2cm

\begin{proof}[\textbf{Proof of Corollary~\ref{tuiasdwqrqacasdwl}}]
Following the lines in the proof of  Lemma \ref{asa0}, it can be verified that under  $b(0)\le f_{C_1,\beta_a,\beta_H,N,d}(\psi_2)$, Assumption~\textbf{A2.b} and Condition \textbf{C1.a},
$F(k)$ is invertible and  $\|G(k)\|\le \psi_2$ \rm{a.s.},~$\forall~k\ge 0$.

Noting that $\mathcal F(mh-1)\subseteq\mathcal F(k-1),k\ge mh$,
by    the properties of the conditional expectation, we have
\bnaa\label{a3esda2easd}
&&~~~\mathbb E[\overline A(k,q)[[\Phi_F(k-1,k-q)]^{-1}-I_{Nn}]|\mathcal F(mh-1)]\cr
&&=\mathbb E[\mathbb E[\overline A(k,q)[[\Phi_F(k-1,k-q)]^{-1}-I_{Nn}]|\mathcal F(k-1)]|\mathcal F(mh-1)]\cr
&&=\mathbb E[\mathbb E[\overline A(k,q)|\mathcal F(k-1)][[\Phi_F(k-1,k-q)]^{-1}-I_{Nn}]|\mathcal F(mh-1)].
\enaa
Since  $\{ \langle\mathcal H(k),\mathcal A_{\mathcal G(k)},\lambda_{ji}(k),j,i\in\mathcal V \rangle, k\ge0\}$ is  an independent process, by Assumption~\textbf{A1.a}, we know that $\overline A(k,q)$ is independent of $\mathcal F(k-1),q=0,...,d$. Then, by (\ref{a3esda2easd}), we have
\bnaa\label{a3easiidasaasd}
&&~~~\mathbb E[\overline A(k,q)[[\Phi_F(k-1,k-q)]^{-1}-I_{Nn}]|\mathcal F(mh-1)]\cr
&&=\mathbb E[\mathbb E[\overline A(k,q)][[\Phi_F(k-1,k-q)]^{-1}-I_{Nn}]|\mathcal F(mh-1)]\cr
&&=\mathbb E[\overline A(k,q)]\mathbb E[[[\Phi_F(k-1,k-q)]^{-1}-I_{Nn}]|\mathcal F(mh-1)],\cr
&&~~~~~~~~~~~~k=mh,...,(m+1)h-1,q=0,...,d.
\enaa

Let $\overline{G}_q(k)=I_{Nn}-\Phi_F(k-1,k-q),q=0,...,d$. Then, $\Phi_F(k-1,k-q)=I_{Nn}-\overline{G}_q(k)$. Noting that $\|{G}(k)\|\leq\psi_2<2^\frac{1}{d}-1,$ by  the  binomial expansion,
we have $\|\overline{G}_q(k)\|=\|I_{Nn}-(I_{Nn}-G(k-1)\cdots (I_{Nn}-G(k-q)\|\leq[(1+\psi_2)^q-1]<1$.
Hence, $[\Phi_F(k-1,k-q)]^{-1}=(I_{Nn}-\overline{G}_q(k))^{-1}=\sum_{i=0}^\infty \overline{G}_q^i(k).$ It follows that $[\Phi_F(k-1,k-q)]^{-1}-I_{Nn}=\sum_{i=1}^\infty \overline{G}_q^i(k)$.
Therefore,
\bnaa\label{926hao}
\|[\Phi_F(k-1,k-q)]^{-1}-I_{Nn}\|&\leq&\Big\|\sum_{i=1}^\infty \overline{G}_q^i(k)\Big\|\le \sum_{i=1}^\infty [(1+\psi_2)^q-1]^i\cr
&=&\frac{(1+\psi_2)^q-1}{2-(1+\psi_2)^q}, q=0,...,d~\mathrm{a.s.}
\enaa

Noting that  for any symmetric matrix~$B\in\mathbb R^{n\times n}$,~$B\ge \lambda_{\min}(B)I_n,~B\le \|B\|I_n$, and for any  matrix~$B\in\mathbb R^{n\times n}$,   $\|B\|=\|B^T\|,$ by the definition of $\overline{\Lambda}_m^h$, we have
\bnaa\label{asijdda9a90wks}
&&\sum_{k=mh}^{(m+1)h-1}\Bigg({b(k)}\mathbb E[\widehat{\mathcal L}_{\mathcal G(k)}]\otimes I_n+a(k)\mathbb E[{\mathcal H}^T(k){\mathcal H}(k)]\cr
&&-\frac{b(k)}{2}\sum_{q=0}^d
\mathbb E[\overline A(k,q)]\mathbb E[[[\Phi_F(k-1,k-q)]^{-1}-I_{Nn}]|\mathcal F(mh-1)]\cr
&&-\frac{b(k)}{2}\sum_{q=0}^d\mathbb E[[[\Phi_F(k-1,k-q)]^{-1}-I_{Nn}]|\mathcal F(mh-1)]^T\mathbb E[\overline A^T(k,q)]\Bigg)\cr
&&\ge\overline{\Lambda}_m^hI_{Nn}-\Bigg\|\sum_{k=mh}^{(m+1)h-1}{b(k)}\sum_{q=0}^d
\mathbb E[\overline A(k,q)]\mathbb E[[\Phi_F(k-1,k-q)]^{-1}-I_{Nn}|\mathcal F(mh-1)]\Bigg\|I_{Nn},\cr
&&
\enaa
By the above, (\ref{a3easiidasaasd}), (\ref{926hao}) and the definition of  ${\widetilde\Lambda_m^h}$, we have
\bann\label{lamsadasd}
{\widetilde\Lambda_m^h}
&\ge&\overline{\Lambda}_m^h-\Bigg\|\sum_{k=mh}^{(m+1)h-1}{b(k)}\sum_{q=0}^d
\mathbb E[\overline A(k,q)]\mathbb E[[\Phi_F(k-1,k-q)]^{-1}-I_{Nn}|\mathcal F(mh-1)]\Bigg\|\cr
&\ge&\overline{\Lambda}_m^h-\sum_{k=mh}^{(m+1)h-1}{b(k)}\sum_{q=0}^d
\|\mathbb E[\overline A(k,q)]\|\frac{(1+\psi_2)^q-1}{2-(1+\psi_2)^q}\ge c(m)
\eann
where the last inequality follows by the condition (\ref{926hao123}). Hence, ${\widetilde\Lambda_m^h}\ge c(m)$.  By Theorem \ref{e4t5634332} and the conditions of the corollary, the proof is completed.
\end{proof}

\vskip 0.2cm

\begin{proof}[\textbf{Proof of Corollary~\ref{tuilunasdasaw}}]
We first prove the first part of the corollary. Let $c(m)= \min\{a((m+1)h),b((m+1)h)\}.$  Since $\{\mathcal G(k),k\ge0\}\in\Gamma_1$,  we know that $\mathbb E[{\widehat{\mathcal L}}_{\mathcal G(k)}|\mathcal F(mh-1]$ is positive semi-definite, $k\ge mh$. Then, by the definitions of $\overline{\Lambda}_m^h$ and ${{\Lambda_m^h}}$, we have
\bnaa\label{saidh982q3riwrf}
\overline{\Lambda}_m^h\ge c(m){{\Lambda_m^h}}.
\enaa
Then, noting that $c(m)\ge\min\{1,1/C_1\}a((m+1)h)$, by the definitions of $C_2$ and $C_3$, we have
\bnaa\label{a9jjiiiiiiis}
b(mh)\le C_2a(mh)\le C_2 (C_3)^ha((m+1)h)\le  C_2 (C_3)^h\max\{1,C_1\}c(m).
\enaa
By the definitions of ${\widetilde\Lambda_m^h}$ and  $\Sigma_m^h$,  (\ref{saidh982q3riwrf})  and (\ref{a9jjiiiiiiis}), similar to (\ref{asijdda9a90wks}),
 we have
\bann\label{sadasdqq3435af}
&&~~~{\widetilde\Lambda_m^h}\cr
&&\ge\overline{\Lambda}_m^h-\sum_{k=mh}^{(m+1)h-1}{b(k)}\Bigg(\sum_{q=0}^d
\|\mathbb E[\overline A(k,q)([\Phi_F(k-1,k-q)]^{-1}-I_{Nn})|\mathcal F(mh-1)]\|\Bigg)\cr
&&\ge\overline{\Lambda}_m^h-b(mh)\sum_{k=mh}^{(m+1)h-1}\Bigg(\sum_{q=0}^d
\|\mathbb E[\overline A(k,q)([\Phi_F(k-1,k-q)]^{-1}-I_{Nn})|\mathcal F(mh-1)]\|\Bigg)\cr
&&\ge c(m){{\Lambda_m^h}}-c(m)\Sigma_{m}^h\ge c(m)\theta~\mathrm{a.s.}
\eann
where  $\theta>0$  by the condition (\ref{ssaoidjaoi90e0}). By Conditions \textbf{C1.a} and  \textbf{C1.b}, similarly to (\ref{aspifjpazdf09})-(\ref{aspifjpazdf093}),  it follows that  $\sum_{m=0}^\infty c(m)=\infty$ and $b^2(mh)=o(c(m))$. Then, the algorithm (\ref{asapp}) converges in mean square by Theorem \ref{e4t5634332}.

We next prove the second part of the corollary.
Since  $\{ \langle\mathcal H(k),\mathcal A_{\mathcal G(k)},\lambda_{ji}(k),j,i\in\mathcal V \rangle, k\ge0\}$ is  an independent process, by  (\ref{a3easiidasaasd}) and (\ref{926hao}), we have
\bann\label{a3easadwe23423}
&&~~~\|\mathbb E[\overline A(k,q)[[\Phi_F(k-1,k-q)]^{-1}-I_{Nn}]|\mathcal F(mh-1)]\|\cr
&&=\|\mathbb E[\overline A(k,q)]\mathbb E[[[\Phi_F(k-1,k-q)]^{-1}-I_{Nn}]|\mathcal F(mh-1)]\|\cr
&&\le \|\mathbb E[\overline A(k,q)]\|\frac{(1+\psi_2)^q-1}{2-(1+\psi_2)^q},q=0,...,d.
\eann
 Noting the definition of $\Sigma_m^h$, we then have
\bann
\Sigma_m^h\le C_2 (C_3)^h\max\{1,C_1\} \sup_{m\ge0}\sum_{k=mh}^{(m+1)h-1}\Bigg(\sum_{q=0}^d \|\mathbb E[\overline A(k,q)]\|\frac{(1+\psi_2)^q-1}{2-(1+\psi_2)^q}\Bigg).
\eann
By the above and  the condition (\ref{asdasdqwra}), we know that $
\inf_{m\ge0}(\Lambda_m^h-\Sigma_m^h)\ge\theta$
where
$$\theta\triangleq\inf_{m\ge0}\Lambda_m^h-C_2 (C_3)^h\max\{1,C_1\}\sup_{m\ge0}\sum_{k=mh}^{(m+1)h-1}\Bigg(\sum_{q=0}^d \|\mathbb E[\overline A(k,q)]\|\frac{(1+\psi_2)^q-1}{2-(1+\psi_2)^q}\Bigg)>0.$$
Then, the proof is completed.
\end{proof}

\vskip 0.2cm

\begin{proof}[\textbf{Proof of Corollary~\ref{tuilunl}}]
 Following the lines of  the proof of Lemma \ref{asa0}, it can be verified that by  $b(0)\le f_{C_1,\beta_a,\beta_H,N,d}(\psi_3)$, Assumption~\textbf{A2.b} and Condition \textbf{C1.a},
$F(k)$ is invertible~\rm{a.s.} and $\|G(k)\|\le \psi_3$~\rm{a.s.},~$\forall~k\ge 0$.

Let $c(m)=\min\{a((m+1)h),b((m+1)h)\}.$ Recalling the definition of $\Sigma_m^h$  in Corollary \ref{tuilunasdasaw}, by (\ref{saidh982q3riwrf}) and (\ref{a9jjiiiiiiis}),
 we have
\bnaa\label{sadasoffujaf}
{\widetilde\Lambda_m^h}&\ge&{\overline\Lambda_m^h}-\sum_{k=mh}^{(m+1)h-1}{b(k)}\Bigg(\sum_{q=0}^d
\|\mathbb E[\overline A(k,q)([\Phi_F(k-1,k-q)]^{-1}-I_{Nn})|\mathcal F(mh-1)]\|\Bigg)\cr
&\ge& c(m)({{\Lambda_m^h}}-\Sigma_m^h)\ge c(m)(\theta-\Sigma_m^h),
\enaa
where   the last inequality follows by $\inf_{m\ge0}{{\Lambda_m^h}}\ge \theta~\rm{a.s.}$
We next prove that $\theta-\Sigma_m^h$ has a  positive lower bound under the conditions of the corollary.

By the definition of $\psi_3$, similar to (\ref{926hao}), we have
\bann\label{sahsaidjqosmd}
\|[\Phi_F(k-1,k-q)]^{-1}-I_{Nn}\|&\leq&\frac{(1+\psi_3)^q-1}{2-(1+\psi_3)^q}, q=0,...,d~\mathrm{a.s.}
\eann
By the above, we have
\bann
\frac{\Sigma_m^h}{C_2 (C_3)^h\max\{1,C_1\}}&=&\sum_{k=mh}^{(m+1)h-1}\sum_{q=0}^d
\|\mathbb E[\overline A(k,q)([\Phi_F(k-1,k-q)]^{-1}-I_{Nn})|\mathcal F(mh-1)]\|\cr
&\le&\sum_{k=mh}^{(m+1)h-1}\sum_{q=0}^d
\mathbb E[\|\overline A(k,q)\|\|[\Phi_F(k-1,k-q)]^{-1}-I_{Nn}\||\mathcal F(mh-1)]\cr
&\le&\sum_{k=mh}^{(m+1)h-1}\sum_{q=0}^d
\mathbb E[\|\overline A(k,q)\||\mathcal F(mh-1)]\frac{(1+\psi_3)^q-1}{2-(1+\psi_3)^q}\cr
&\le&\frac{(1+\psi_3)^d-1}{2-(1+\psi_3)^d}\sum_{k=mh}^{(m+1)h-1}\Bigg(\sum_{q=0}^d
\mathbb E[\|\overline A(k,q)\||\mathcal F(mh-1)]\Bigg)\cr
&\le&{N}\beta_adh\frac{(1+\psi_3)^d-1}{2-(1+\psi_3)^d}.
\eann
This together with $\psi_3<\Bigg(1+\frac{\theta}{\theta+{N}C_2 (C_3)^h\max\{1,C_1\}\beta_adh}\Bigg)^{\frac{1}{d}}-1$ gives
$$\theta-\Sigma_m^h\ge\theta-{N}C_2 (C_3)^h\max\{1,C_1\}\beta_adh\frac{(1+\psi_3)^d-1}{2-(1+\psi_3)^d}>0.$$
Then, by (\ref{sadasoffujaf}), we have ${\widetilde\Lambda_m^h}\ge c'(m)$, $m\ge0$,
where $$c'(m)=c(m)\Bigg[\theta-{N}C_2 (C_3)^h\max\{1,C_1\}\beta_adh\frac{(1+\psi_3)^d-1}{2-(1+\psi_3)^d}\Bigg].$$
Similarly to (\ref{aspifjpazdf09})-(\ref{aspifjpazdf093}), by Conditions \textbf{C1.a} and \textbf{C1.b}  it is known that $\sum_{m=0}^\infty c'(m) =\infty$ and $b^2(mh)=o(c'(m))$.  By Theorem \ref{e4t5634332}, we get the conclusion of the corollary.
\end{proof}

 \end{appendices}
\renewcommand{\appendixname}{Appendix~D: The deterministic observation matrices in the simulation}
\appendix
$H_1'=[\widetilde H_{1},\mathbf{0}_{5\times9}],H_2'=[\widetilde H_{2},\mathbf{0}_{7\times5}],H_3'=[\mathbf{0}_{6\times4},\widetilde H_3],H_4'=[\mathbf{0}_{4\times7},\widetilde H_4]$, where
\bann
&\widetilde H_{1}=\left(
         \begin{array}{cccc}
           -1 & 0 &0 & 0 \\
           0& 0 & 0& -1 \\
           1 & 0& 0 & -1\\
           -1 & 0& 0&-1\\
           -1 & 0& -1&3\\
         \end{array}
       \right),
\widetilde H_2=\left(
             \begin{array}{cccccccc}
               0 & 0& 0& 0& 0& -1& 1&0\\
               0 &0 &-1& 0 & 0 & 1& 0 &0\\
               0 & 1 & -1 & 0 & 0 & 0 & 0&0 \\
                0 & 1 & -1 & 0 & 0 & 0 & 0&0 \\
                0 & 0& 1 & 0 & 0 & 0 & 0&-1 \\
                0 & 0& 1 & 0 & 0 & 1 & 0&-1 \\
                  0 & 0 & 1 & -1 & 0 & 0 & 0&0 \\
             \end{array}
           \right)\cr
&\widetilde H_3=\left(
             \begin{array}{ccccccccc}
               1 & 0 & 0 & 0& 0 & 0 & -1& 0 &0\\
               1 & 0 &0 & 0 & 0 & 0 & 0 & -1 &0\\
               0 &0 & 0 & 0 & 0& 0 & 1& -1 &0\\
               -1 & 0 &0 & 0 & 0 & 0 & 2 & 1&0 \\
               -1 & 0 &0 & 0 & 0 & 0 & -1& 3&-1 \\
               0 & 0 &0 & 0 & 0 & 0 & 0& 1&-1 \\
             \end{array}
           \right),
\widetilde H_4=\left(
              \begin{array}{cccccc}
                1 &-1& 0& 0 & 0 & 0 \\
                1 & 0 & 0 & 0& 0 & -1 \\
                -1 & 0 & 0 & 0& -1 & 2 \\
                0 & 1 & -1 & 0& 0 & 0\\
              \end{array}
            \right).
\eann

\setlength{\baselineskip}{12pt}
\end{document}